\documentclass[final]{siamltex}

\usepackage{amsmath,amsxtra,amssymb,latexsym,amscd} 
\usepackage{empheq}
\usepackage{graphicx}
\usepackage{calligra}
\usepackage[capposition=top]{floatrow}
\usepackage[titletoc,title]{appendix}

\DeclareMathOperator{\Ima}{Im}

\newtheorem{remark}{Remark}[section]

\def\L{{\mathcal T}}
\def\TT{{\bf T}}
\def\VV{{\bf V}}

\def\rank{{\rm rank}}

\title{Mean exit time for surface-mediated diffusion: spectral analysis and asymptotic behavior}

\author{O. B\'enichou\footnotemark[1] 
\and D.~S. Grebenkov\footnotemark[2]
\and L. Hillairet\footnotemark[3]
\and L. Phun\footnotemark[3]
\and R. Voituriez\footnotemark[1]
\and M. Zinsmeister\footnotemark[3] \ \footnotemark[4]}

\begin{document}

\maketitle

\renewcommand{\thefootnote}{\fnsymbol{footnote}}
\footnotetext[1]{Laboratoire de Physique Th\'eorique de la Mati\`ere Condens\'ee (UMR 7600), CNRS / UPMC, 4 Place Jussieu, 75255 Paris Cedex}
\footnotetext[2]{Laboratoire de Physique de la Mati\`ere Condens\'ee (UMR 7643) CNRS -- Ecole Polytechnique, 91128 Palaiseau, France}
\footnotetext[3]{MAPMO (UMR 6628), Universit\'e d'Orl\'eans, 45067 Orl\'eans, France}
\footnotetext[4]{Corresponding author: zins@univ-orleans.fr}
\renewcommand{\thefootnote}{\arabic{footnote}}

\date{Received: \today / Revised version: }

\begin{abstract}
We consider a model of surface-mediated diffusion with alternating
phases of pure bulk and surface diffusion.  For this process, we compute
the mean exit time from a disk through a hole on the circle.  We
develop a spectral approach to this escape problem in which the mean
exit time is explicitly expressed through the eigenvalues of the
related self-adjoint operator.  This representation is particularly
well suited to investigate the asymptotic behavior of the mean exit
time in the limit of large desorption rate $\lambda$.  For a
point-like target, we show that the mean exit time diverges as
$\sqrt{\lambda}$.  For extended targets, we establish the asymptotic
approach to a finite limit.  In both cases, the mean exit time is
shown to asymptotically increase as $\lambda$ tends to infinity.  We
also revise the optimality regime of surface-mediated diffusion.
Although the presentation is limited to the unit disk, the spectral
approach can be extended to other domains such as rectangles or
spheres.
\end{abstract}

\section{Introduction}

Many transport and search processes exhibit intermittent character
when different modes of motion are alternated.  Typical examples are
animals foraging (with phases of rapid relocation and slow
exploration), facilitated search mechanism on DNA (with phases of pure bulk
diffusion and chain sliding), vesicle transportation in living cells
(with phases of active transport by motor proteins and passive
diffusion in the cytoplasm), water transport in confining media (with
phases of pure bulk diffusion and surface exploration)
\cite{Benichou11,Bressloff13}.  The intermittence is often expected to
facilitate transport and search processes, e.g., by reducing the mean
search time necessary to reach a target (food, specific DNA sequence,
nucleus, or reaction zone in the above examples).  In particular, the
mean exit time from a bounded domain through an opening (a target) on
the boundary has been actively studied during the last decade
\cite{Redner,Metzler}.  For pure bulk diffusion, Singer {\it et al.}
derived the asymptotic behavior of the mean exit time in the narrow
escape limit (when the size of the target is small)
\cite{Singer06a,Singer06b,Singer06c,Schuss07,Reingruber09}.  Isaacson
and Newby developed uniform in time asymptotic expansions in the
target radius of the first passage time density for the diffusing
molecule to find the target \cite{Isaacson13}.  The escape problem for
an intermittent process with phases of surface and pure bulk diffusion (the
so-called surface-mediated diffusion) has been recently solved for
rotation-invariant domains \cite{Benichou10,Benichou11a}.  The known
eigenbases for the Laplace operators governing pure bulk and surface
diffusions allowed one to express the mean exit time in a closed
matrix form.  Under well-defined conditions, the mean exit time was
shown to be minimized at an optimal desorption rate that characterizes
switching from surface to pure bulk diffusion.  These results have been
extended in various directions
\cite{Rupprecht12a,Rupprecht12b,Calandre12,Calandre14}.  An
alternative master equation approach for discrete (on-lattice)
surface-mediate diffusion (also called the bulk-mediated surface
diffusion) has been proposed
\cite{Revelli05,Rojo11a,Rojo11b,Rojo13}.

In the present paper, we propose a rigorous spectral analysis of the
above escape problem.  We focus on surface-mediated diffusion in the
unit disk and derive a spectral representation of the mean exit time.
This representation is well suited to investigate the asymptotic
behavior of the mean exit time in the limit of large desorption rate
$\lambda$.  For a point-like target, we show that the mean exit time
diverges as $\sqrt{\lambda}$.  For extended targets, we establish the
asymptotic approach to a finite limit.  In both cases, the mean exit
time is shown to asymptotically increase as $\lambda$ tends to
infinity.  We also revise the optimality regime of surface-mediated
diffusion.  Although the presentation is limited to the unit disk, the
spectral approach can be extended to other domains such as rectangles
or spheres.

\section{A self-adjoint operator formulation}

We study the following model of surface-mediated diffusion in the unit
disk $\mathbb{D} = \{ z\in \mathbb{C}~: |z|<1\}$ whose boundary
$\partial \mathbb{D}$ includes an exit (or a target) of angular size
$2\epsilon$ (i.e., an arc of the unit circle between $\pi-\epsilon$
and $\pi+\epsilon$), with $0\leq \epsilon \leq \pi$.  A starting point
$e^{i\theta}$ is taken on the unit circle.  If the starting point is
located on the target then the process is immediately stopped.
Otherwise, the particle moves along the circle according to a Brownian
motion with the diffusion coefficient $D_1$ for a duration of
$\min\{\tau_\lambda,\tau\}$, where $\tau_\lambda$ is a random variable
with exponential law of parameter $\lambda \geq 0$, and $\tau$ is the
first hitting time of the target.  If $\tau \leq \tau_\lambda$ then
the process stops.  If $\tau > \tau_\lambda$ then the particle is
relocated at time $\tau_\lambda$ along the normal inside the disk at a
distance $0 < a \leq 1$ to start there a 2D Brownian motion with the
diffusion coefficient $D_2$.  This motion is stopped after hitting
back the unit circle, and the same procedure is restarted from this
last hitting point.  We define $t_1(\theta)$ as being the expected
time to reach the target.  Similarly, for $0 \le r < 1$, we define
$t_2(re^{i\theta})$ as being the expected time to reach the target
starting from the point $re^{i\theta}$ inside the unit disk.

It has been shown in \cite{Benichou11a} that these two functions
satisfy the following system of equations:
\begin{empheq}[left=\empheqlbrace]{align}
&D_1 t_1'' (\theta) + \lambda [t_2\left((1-a)e^{i\theta}\right)-t_1(\theta)]=-1 \label{eq:basic_eq_1} \\
&D_2\Delta t_2=-1    \label{eq:basic_eq_2}\\
&t_2(e^{i\theta})=t_1(\theta) \text{ } (\theta \in [-\pi,\pi]) \label{eq:boundary_condition_1}\\
&t_1(\theta)=0 \text{ if } \theta \in [-\pi,-\pi+\epsilon] \cup [\pi-\epsilon, \pi] \label{eq:boundary_condition_2}
\end{empheq}
Let us notice that, by symmetry, $t_1(\theta)$ is an even function so
it is sufficient to determine it on $[0,\pi]$.

The solution to Eq. \eqref{eq:basic_eq_2} is the sum of the particular
solution $\frac{1-r^2}{4D_2}$ to the inhomogeneous (Poisson) equation
$\Delta u= -\frac{1}{D_2}, ~ u|_{\partial \mathbb{D}}=0$, and the
solution to the Dirichlet problem $\Delta v=0$, $v|_{\partial
\mathbb{D}}=t_1$.

Since $t_1$ is even it may be represented as a cosine series
\begin{equation*}
t_1(\theta)=\sum_{n \ge 0}{a_n\cos n\theta},
\end{equation*}
from which
\begin{align*}
t_2(re^{i\theta})=\frac{1-r^2}{4D_2}+\sum_{n\geq 0}{a_n r^n\cos n\theta}.
\end{align*}
Eq. \eqref{eq:basic_eq_1} then becomes
\begin{align}
\label{eq:t1_diff}
t_1''(\theta) = -\frac{1}{D_1}\left(1 + \lambda\frac{1-(1-a)^2}{4D_2}\right) + \frac{\lambda}{D_1} U(t_1) ,
\end{align}
where $U$ is the operator on $L^2\left( [0,\pi] \right)$ defined by
\begin{align*}
U\left( \sum_{n \ge 0}{x_n \cos n \theta} \right)=\sum_{n \ge 1}{x_n \left(1-(1-a)^n\right)\cos n\theta}.
\end{align*}
This operator can also be written as $U=V^2$, where
\begin{align}\label{eq:operator_V_A}
V\left( \sum_{n \ge 0}{x_n \cos n \theta} \right)=\sum_{n \ge 1}{x_n \sqrt{1-(1-a)^n}\cos n\theta}.
\end{align}

Next we introduce the Sturm Liouville operator $T$ defined on
$L^2\left([0,\pi-\epsilon]\right)$ as $Tf=u$, where
\begin{empheq}[left=\empheqlbrace]{align}\label{eq:eq_T_tilde}
& u''=f\\
& u'(0)=u(\pi-\epsilon)=0.
\end{empheq}
The operator $T$ is negative self-adjoint. Finally, we define
$\tilde{T} = -ETR$ as an operator on $L^2\left( [0,\pi] \right)$,
where $R: L^2\left([0,\pi]\right) \longrightarrow
L^2\left([0,\pi-\epsilon]\right)$ is the natural restriction, and
$E:L^2\left([0,\pi-\epsilon]\right) \longrightarrow
L^2\left([0,\pi]\right)$ is the natural extension by $0$.  The
operator $\tilde{T}$ can be written explicitly as
\begin{equation}
\tilde{T}(f) = \begin{cases} \int\limits_{\theta}^{\pi -\epsilon} d\theta_1 \int\limits_0^{\theta_1} d\theta_2 ~ f(\theta_2) , 
\quad 0\leq \theta < \pi -\epsilon ,\cr
0 ,\hskip 31mm   \pi - \epsilon \leq \theta \leq \pi .\end{cases}
\end{equation}
One can easily check that the eigenbasis of this operator is
\begin{equation}
\nu_n = \frac{(1-\epsilon/\pi)^2}{(n+1/2)^2},  \qquad
u_n = \begin{cases}
 \sqrt{\frac{2}{\pi-\epsilon}} \cos \left(\frac{(n+1/2)\theta}{1-\epsilon/\pi}\right), \quad 0 \le \theta \le \pi - \epsilon, \cr
 0, \hskip 35mm \pi-\epsilon \le \theta \le \pi.   \end{cases}
\end{equation}
These eigenvectors form an orthogonal basis of $L^2[0,\pi-\epsilon]$
by Sturm-Liouville theory.

Let us check that $\tilde{T}$ is a self-adjoint non-negative operator
on $L^2\left([0,\pi]\right)$.  Let $f$, $g$ be in
$L^2\left([0,\pi]\right)$:
\begin{align*}
- \langle \tilde{T} f, \overline{g} \rangle & = \langle E T R f, \overline{g} \rangle = \int_{0}^{\pi-\epsilon} T(Rf) ~\overline{g}  
=\int_{0}^{\pi-\epsilon} T(Rf)~ \overline{Rg}\\
&=\int_{0}^{\pi-\epsilon} Rf~ \overline{T(Rg)}
=\int_{0}^{\pi}f~\overline{E T R(g)} 
= - \langle f, \overline{\tilde{T}g} \rangle ,
\end{align*}
which proves the claim, since the operator $ETR$ is negative.

The operator $\tilde{T}$ allows us to translate
Eq. \eqref{eq:t1_diff} into
\begin{align}\label{eq:t1}
t_1=\frac{1}{D_1}\left(1+\lambda\frac{1-(1-a)^2}{4D_2}\right)\tilde{T}(1)-\frac{\lambda}{D_1}\tilde{T}U(t_1).
\end{align}
We next apply the operator $V$ to both sides of Eq. \eqref{eq:t1} to
get, writing $s_1=V(t_1)$,
\begin{align*}
s_1=\frac{1}{D_1}\left(1+\lambda\frac{1-(1-a)^2}{4D_2}\right)V\tilde{T}(1)-\frac{\lambda}{D_1}V\tilde{T}V(s_1) ,
\end{align*}
which can be solved in $s_1$ as
\begin{align}\label{eq:s1}
s_1 = \frac{1}{D_1}\left(1+ \lambda \frac{1-(1-a)^2}{4D_2} \right)\left(I+\frac{\lambda}{D_1}V\tilde{T}V\right)^{-1}(\psi), 
\end{align}
where $\psi=V\tilde{T}(1)$.  This is an exact solution of the original
problem for a fixed starting point.  We emphasize that the operators
$V$ and $\tilde{T}$, as well as the function $\psi = V\tilde{T}(1)$,
are given explicitly.  At first thought, this representation looks
similar to the mean exit time found in \cite{Benichou11a} (see also
\cite{Rupprecht12a,Rupprecht12b}).  Although both derivations are
conceptually similar, the major advantage of the present approach is
the use of the self-adjoint operator $V\tilde{T}V$.  This feature
allows one to invert the operator $(I + \frac{\lambda}{D_1}
V\tilde{T}V)$ in Eq. \eqref{eq:<t1>0} and to express the mean exit
time in a spectral form (see below).

The case of a randomly distributed starting point on the circle with
uniform law is of particular interest.  This is equivalent to
averaging the mean exit time over the starting points that we denote
as
\begin{align*}
\langle t_1 \rangle = \frac{1}{\pi} \int_0^\pi d\theta ~ t_1(\theta) = \frac{1}{\pi} \langle t_1,1 \rangle.
\end{align*}
Using Eqs. \eqref{eq:t1} and \eqref{eq:s1}, we can write
\begin{align}\label{eq:<t1>}
\pi \langle t_1 \rangle &=\frac{1}{D_1}\left(1+\lambda\frac{1-(1-a)^2}{4D_2}\right)\langle \tilde{T}(1),1 \rangle -\frac{\lambda}{D_1} 
\langle \tilde{T}V(s_1),1 \rangle \nonumber \\
&=\frac{1}{D_1}\left(1+\lambda\frac{1-(1-a)^2}{4D_2}\right)\langle \tilde{T}(1),1 \rangle -\frac{\lambda}{D_1} \langle s_1,\psi \rangle , 
\end{align}
from which it follows that the knowledge of $s_1$ allows to compute
$\langle t_1 \rangle$:
\begin{align}\label{eq:<t1>0}
\langle t_1 \rangle &=\frac{1}{\pi D_1}\left(1+\lambda\frac{1-(1-a)^2}{4D_2}\right)\left(\langle \tilde{T}(1),1 \rangle - 
\frac{\lambda}{D_1} \langle \left(I+\frac{\lambda}{D_1}V\tilde{T}V\right)^{-1} \psi , \psi \rangle \right).
\end{align}

By spectral theorem there exists an orthonormal basis of $L^2 \left(
[0,\pi] \right)$ which diagonalizes the self-adjoint operator
$V\tilde{T}V$.  More precisely, $L^2 \left( [0,\pi] \right)$ is the
orthogonal direct sum of $\ker (V\tilde{T}V)$ and $\Ima (V\tilde{T}V)$
and we obtain this orthonormal basis by completing any orthonormal
basis of $\ker (V\tilde{T}V)$ with the basis formed by the normalized
eigenvectors associated with positive eigenvalues.

To identify these two spaces let us notice first that $\ker V$ is the
one dimensional space of constant functions. Thus $\ker (V\tilde{T}V)$
is the space of functions $f \in L^2 \left( [0,\pi] \right)$ such that
$\tilde{T}(Vf)$ is constant. But since $Vf=\tilde{T}(Vf)''$ on $[0,
\pi-\epsilon]$, $Vf \equiv 0$ on $[0, \pi-\epsilon]$.  So $f \in \ker
(V\tilde{T}V) \Rightarrow \supp(Vf) \subset [\pi-\epsilon,\pi]$, and
this implication is easily seen to be an equivalence. With a slight
abuse of language, we write $\ker (V\tilde{T}V)=V^{-1}\left(L^2 \left(
\pi-\epsilon,\pi \right)\right)$.  It follows that $\Ima
(V\tilde{T}V)=V\left( L^2 \left( [0,\pi-\epsilon] \right) \right)$.

We call $\left( e_n \right)_{n \ge 0}$ the orthonormal basis of $\Ima
(V\tilde{T}V)$ such that $V\tilde{T}V e_n = \lambda_n e_n$ and
$\lambda_n \downarrow 0$ as $n \to \infty$. 

When $\epsilon = 0$, the eigenbasis $e_n$ is simply formed by cosine
functions, and the analysis is straightforward (see below).  When
$\epsilon > 0$, we first observe that $\psi =
V\tilde{T}(1)=V\tilde{T}(\varphi_0)$ where
\begin{align*}
\varphi_0 =
	\begin{cases}
	1 &\text{ on } [0,\pi-\epsilon[ , \\
	-\frac{\pi-\epsilon}{\epsilon} &\text{ on } [\pi-\epsilon,\pi] .
	\end{cases}
\end{align*}
so that $\int_{0}^{\pi}{\varphi_0}=0$, and thus there exists $\psi_0
\in L^2 \left( [0,\pi] \right)$ such that $\varphi_0=V\psi_0$.
It follows that $\psi \in \Ima(V\tilde{T}V)$ and we can write
\begin{align*}
\psi=\sum_{n \ge 1}{\psi_n e_n} ,
\end{align*}
with coefficients $\psi_n$ forming a sequence in $\ell^2$.  Using this
representation, Eq. \eqref{eq:s1} is formally solved as
\begin{align*}
s_1=\frac{1}{D_1}\left(1+\lambda\frac{1-(1-a)^2}{4D_2}\right) \sum_{n\geq 1}{\frac{\psi_n}{1+\frac{\lambda}{D_1}\lambda_n}~e_n} .
\end{align*}
Plugging this expression into Eq. \eqref{eq:<t1>0} we obtain
\begin{align}\label{eq:<t1>2}
\langle t_1 \rangle = \frac{1}{\pi D_1}\left(1+\lambda\frac{1-(1-a)^2}{4D_2}\right)\left[ \langle \tilde{T}(1),1 \rangle 
-\frac{\lambda}{D_1} \sum_{n\geq 0}{\frac{\psi_n^2}{1+\frac{\lambda}{D_1}\lambda_n}} \right] .
\end{align}
This spectral representation is particularly well suited for the
asymptotic analysis of the mean exit time.

\subsection{Point-like target ($\epsilon = 0$)}

We first consider the case of $\epsilon = 0$.  Although such a target
is not accessible for 2D pure bulk diffusion, it can still be reached
through 1D surface diffusion.  In this case, one easily gets
\begin{align}
\label{eq:Tcostheta}
\begin{cases}
\tilde{T}(\cos n\theta) = -T(\cos n\theta)=\frac{\cos n\theta -(-1)^n}{n^2} \quad (n\ge 1),\\
\tilde{T}(1) = -T(1) = \frac{\pi^2-\theta^2}{2} ,
\end{cases}
\end{align}
so that
\begin{align}\label{eq:formula_of_VTV}
\begin{cases}
V \tilde{T} V(\cos n\theta)=\frac{1-(1-a)^n}{n^2}\cos n\theta  \quad (n \ge 1), \\
V \tilde{T} V(1)=0 .
\end{cases}
\end{align}
One concludes that 
\begin{equation}
\label{eq:eigen_eps0}
\lambda_n = \begin{cases} \frac{1-(1-a)^n}{n^2}  \quad (n \geq 1), \cr   0 \hskip 15mm (n=0), \end{cases}  \qquad
e_n = \begin{cases}       \sqrt{2/\pi} \cos n\theta  \quad (n\geq 1), \cr  \sqrt{1/\pi} \hskip 14mm (n = 0).  \end{cases}
\end{equation}
For $n \geq 1$, we have
\begin{align}
\psi_n & = \langle \psi, e_n \rangle = \langle V\tilde{T}(1), e_n \rangle = \sqrt{2/\pi}~ \langle \tilde{T}(1),V(\cos n\theta) \rangle \nonumber \\
\label{eq:psi_n}
&= \sqrt{2/\pi} \sqrt{1-(1-a)^n} ~ \langle \frac{\pi^2-\theta^2}{2} , \cos n\theta \rangle 
= \sqrt{2\pi} \sqrt{1-(1-a)^n}~ \frac{(-1)^{n+1}}{n^2} ,  
\end{align}
while $\langle \psi, 1 \rangle = \langle V\tilde{T}(1), 1 \rangle =
\langle \tilde{T}(1), V1 \rangle = 0$.  Substituting this expression
into Eq. \eqref{eq:<t1>2}, we get
\begin{align}
\langle t_1 \rangle_{\epsilon=0} 
&= \frac{1}{\pi D_1}\left( 1+\lambda\frac{1-(1-a)^2}{4D_2} \right) \left( \langle \tilde{T}(1),1 \rangle 
- \sum_{n \ge 1}{\frac{2\pi \frac{\lambda}{D_1}\left( 1-(1-a)^n \right) }{n^2\left( n^2+\frac{\lambda}{D_1}\left( 1-(1-a)^n \right) \right)}} \right) \nonumber\\
&= \frac{1}{\pi D_1}\left( 1+\lambda\frac{1-(1-a)^2}{4D_2} \right)\Biggl( \langle \tilde{T}(1),1 \rangle 
- 2\pi \sum_{n \ge 1}{\frac{1}{n^2}} \nonumber \\
& + 2\pi \sum_{n \ge 1}{\frac{1}{n^2+\frac{\lambda}{D_1}\left(1-(1-a)^n\right)}} \Biggr) . \label{eq:<t1>,epsilon=0}
\end{align}
From Eq. \eqref{eq:Tcostheta}, we get
\begin{align}\label{eq:<t(1),1>}
\langle \tilde{T}(1),1 \rangle_{\epsilon=0} = \int\limits_0^\pi d\theta~ \frac{\pi^2-\theta^2}{2} = \frac{1}{3}\pi^3 .
\end{align}
We also know the value of the Riemann zeta function
\begin{align}\label{eq:zeta_function}
\zeta(2)=\sum_{n \ge 1}{\frac{1}{n^2}}=\frac{\pi^2}{6} .
\end{align}
Plugging Eqs. \eqref{eq:<t(1),1>} and \eqref{eq:zeta_function} into
Eq. \eqref{eq:<t1>,epsilon=0} yields
\begin{align}  \label{eq:meant1_eps0}
\langle t_1 \rangle_{\epsilon=0} &= \frac{2}{D_1} \left( 1+\lambda\frac{1-(1-a)^2}{4D_2} \right)
\sum_{n \ge 1}{\frac{1}{n^2+\frac{\lambda}{D_1}\left(1-(1-a)^n\right)}}.
\end{align}
We retrieved the exact representation of the mean exit time for
point-like target that was first derived in \cite{Benichou10}.

For $0<a<1$, in the following inequalities
\begin{align*}
 \sum_{n \ge 1}{\frac{1}{n^2+\frac{\lambda}{D_1}}} &\le \sum_{n\ge 1}{\frac{1}{n^2+(1-(1-a)^n)\frac{\lambda}{D_1}}}  
\le  \sum_{n\ge 1}{\frac{1}{n^2+a\frac{\lambda}{D_1}}} , 
\end{align*}
Since $f(x) = \frac{1}{x^2+\frac{a\lambda}{D_1}}$ is decreasing, by
applying the property that $\int_{n}^{n+1}{f(x) dx} \le f(n) \le
\int_{n-1}^{n}{f(x) dx}$, one can replace sums by integrals to get
\begin{align*}
\sqrt{\frac{D_1}{\lambda}} \frac{\pi}{2} F(D_1/\lambda) 
& \le \sum_{n\ge 1}{\frac{1}{n^2+(1-(1-a)^n)\frac{\lambda}{D_1}}} \le \sqrt{\frac{D_1}{a\lambda}} 
\frac{\pi}{2},
\end{align*}
where $F(x) = 1 - \frac{2}{\pi} \arctan(\sqrt{x})$.  One obtains
therefore
\begin{align}
\frac{\pi F(D_1/\lambda)}{\sqrt{D_1 \lambda}} \left( 1+\lambda\frac{1-(1-a)^2}{4D_2} \right)  &\le \langle t_1 \rangle_{\epsilon=0} 
\le \frac{\pi}{\sqrt{a D_1 \lambda}}\left( 1+\lambda\frac{1-(1-a)^2}{4D_2} \right) .
\end{align}
In the limit of large $\lambda$, $F(x)$ can be approximated by $1$,
i.e.
\begin{align}  \label{eq:t1_ineq}
\frac{1-(1-a)^2}{4D_2 \sqrt{D_1}} \pi\sqrt{\lambda}+\frac{\pi}{\sqrt{D_1\lambda}} &\lesssim \langle t_1 \rangle_{\epsilon=0} 
\le  \frac{1-(1-a)^2}{4D_2 \sqrt{D_1a}} \pi\sqrt{\lambda}+\frac{\pi}{\sqrt{D_1a\lambda}} .
\end{align}
Therefore, we conclude that $\langle t_1 \rangle_{\epsilon=0}$
asymptotically increases to infinity, as illustrated on
Fig. \ref{fig:t1_eps}a.  As a result, if $\left. \frac{\partial
\langle t_1 \rangle_{\epsilon=0}}{\partial
\lambda}\right|_{\lambda=0} <0$, then $\langle t_1
\rangle_{\epsilon=0}$ has a minimum.  We compute
\begin{align}
&\left. \frac{\partial \langle t_1 \rangle_{\epsilon=0}}{\partial
\lambda}\right|_{\lambda=0} =
\frac{2}{D_1}\left( \frac{1-(1-a)^2}{4D_2}\sum_{n\ge 1}{\frac{1}{n^2}} - \frac{1}{D_1}\sum_{n \ge 1}{\frac{1-(1-a)^n}{n^4}} \right)<0 \nonumber \\
&\Leftrightarrow \frac{1-(1-a)^2}{4D_2}\frac{\pi^2}{6} < \frac{1}{D_1}\sum_{n \ge 1}{\frac{1-(1-a)^n}{n^4}} \nonumber \\
&\Leftrightarrow D_2 > D_{2,\rm crit} = D_1 \frac{\pi^2(1-(1-a)^2)}{24 \sum_{n \ge 1}{\frac{1-(1-a)^n}{n^4}}}
 \label{eq:inequation_define_ratio_D1/D2}
\end{align}
We retrieved the optimality condition first reported in
\cite{Benichou10}.  The relation
\eqref{eq:inequation_define_ratio_D1/D2} determines the critical value
of the pure bulk diffusion coefficient $D_{2,\rm crit}$, which for small
$a$ can be approximated as
\begin{equation}
\lim_{a \rightarrow 0} D_{2,\rm crit} = D_1 \frac{\pi^2}{12 \zeta(3)} \approx 0.68 D_1.
\end{equation}
If $D_2 > D_{2,\rm crit}$, then $\langle t_1 \rangle_{\epsilon=0}$ has
a minimum.

\begin{figure}[H]
\begin{center}
	\includegraphics[width=0.49 \textwidth]{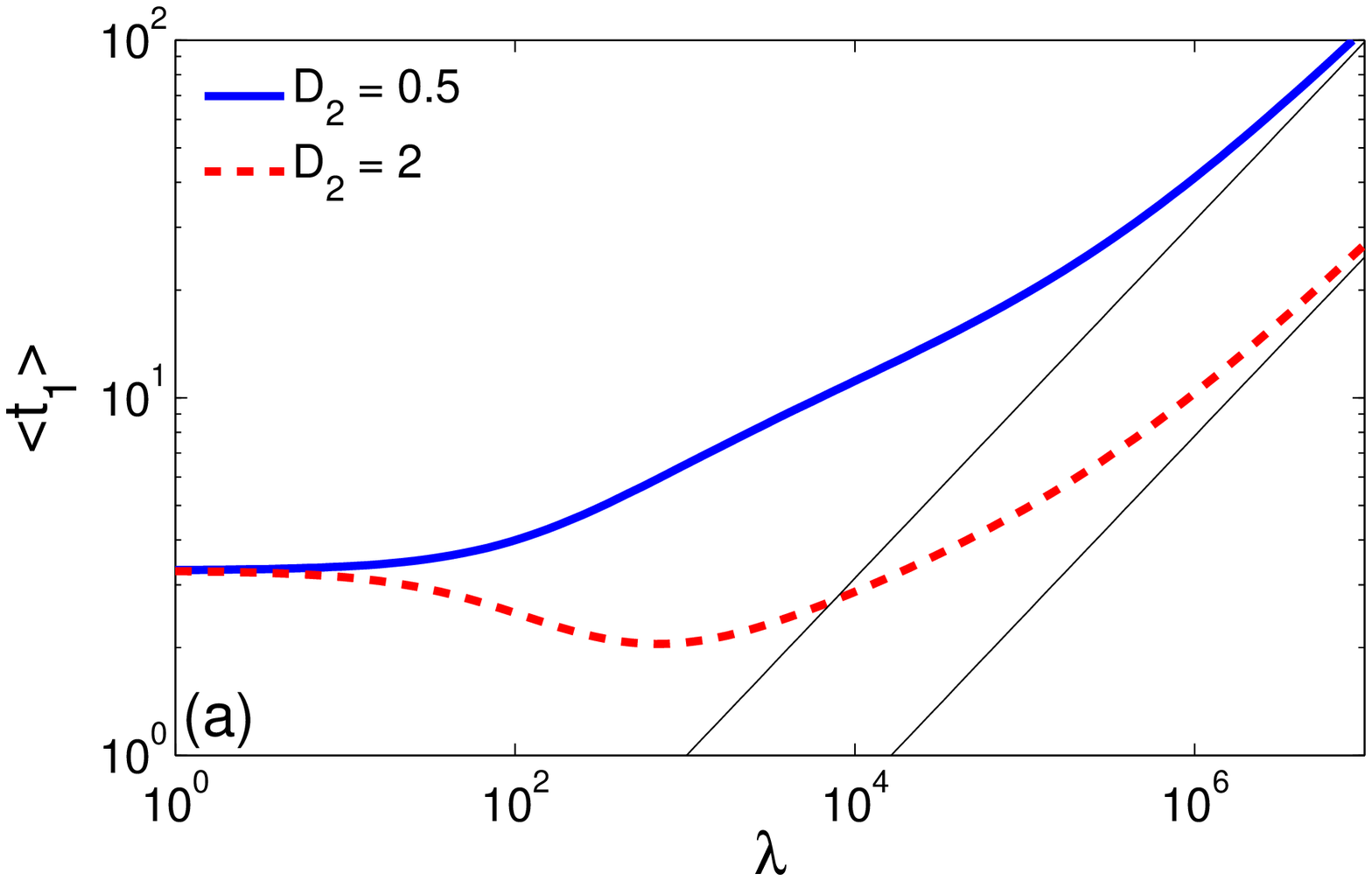}
	\includegraphics[width=0.49 \textwidth]{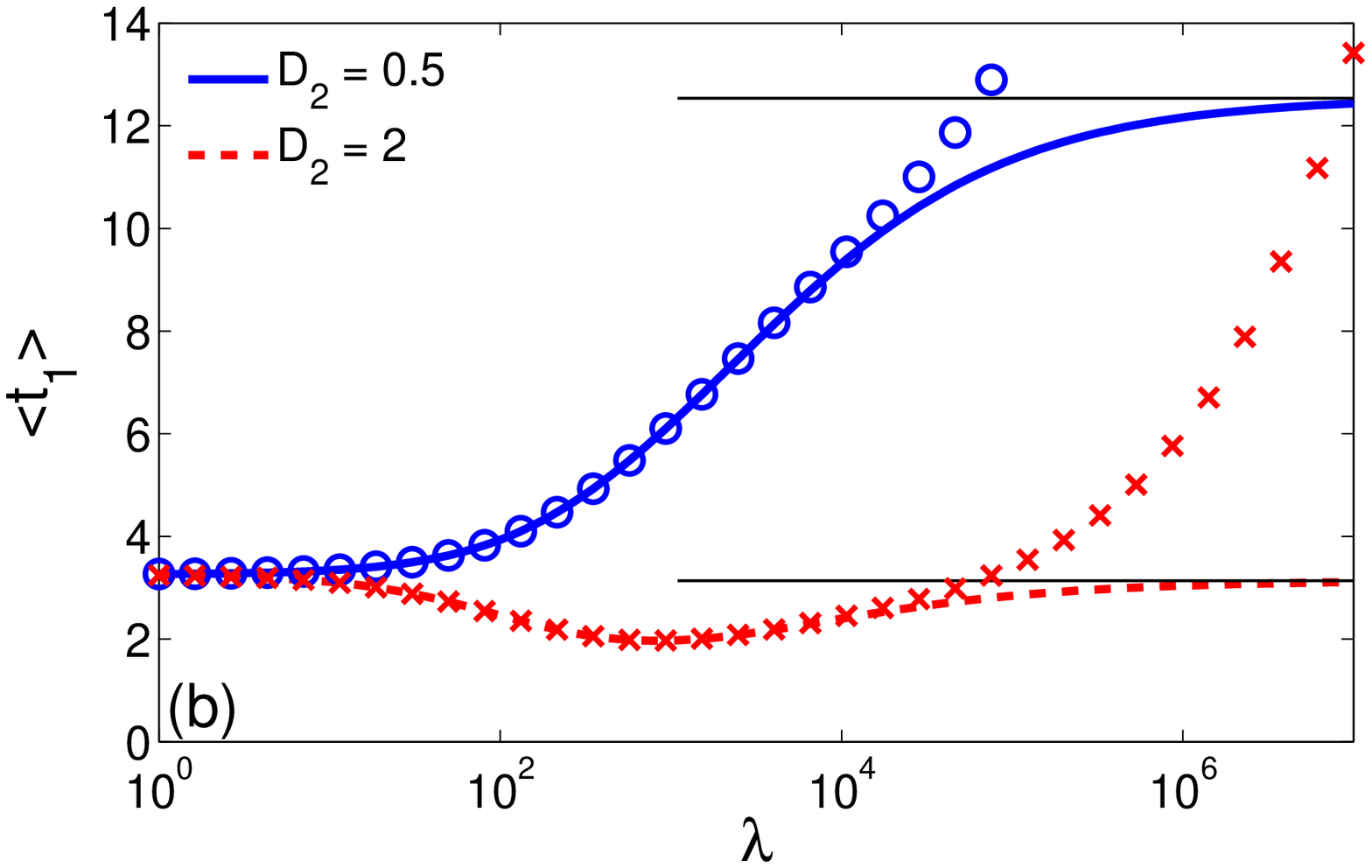}
\end{center}
\caption{
Mean exit time $\langle t_1 \rangle$ as a function of $\lambda$ for
$\epsilon = 0$ ({\bf a}) and $\epsilon = 0.01$ ({\bf b}), with
$a=0.01$ and $D_1=1$.  When $D_2=0.5 < D_{2,\rm crit}$ (blue solid
line), $\langle t_1 \rangle$ monotonously increases with $\lambda$ so
that the smallest mean exit time corresponds to $\lambda = 0$ (surface
diffusion without intermittence).  When $D_2 = 2 > D_{2,\rm crit}$
(red dashed line), $\langle t_1 \rangle$ starts first to decrease with
$\lambda$, passes through a minimum and monotonously increases to
infinity.  For point-like target ({\bf a}), thin solid lines indicate
the leading term of the lower bound $\frac{1-(1-a)^2}{4D_2 \sqrt{D_1}}
\pi \sqrt{\lambda}$ from Eq. \eqref{eq:t1_ineq}.  Note that the upper
bound, which is larger by $a^{-1/2}$, strongly overestimates the mean
exit time.  Finally, correction terms of the order $\lambda^{-1/2}$
are negligible for large $\lambda$.  For extended target ({\bf b}),
horizontal lines indicate the limiting values of the mean exit time as
$\lambda\to\infty$.  Symbols present the diagonal approximation
\eqref{eq:t1_approx}.}
\label{fig:t1_eps}
\end{figure}

\subsection{Extended target ($\epsilon > 0$)}

In sharp contrast to a point-like target, the mean exit time $\langle
t_1 \rangle$ to an extended target ($\epsilon >0$) converges in a
finite limit as $\lambda \rightarrow \infty$.  As a consequence, the
quantity $\langle \tilde{T}(1),1 \rangle -\frac{\lambda}{D_1}
\sum_{n\geq 0}{\frac{\psi_n^2}{1+\frac{\lambda}{D_1}\lambda_n}}$ must
converge to $0$ as $\lambda \rightarrow \infty$.  Given that 
\begin{align*}
\lim_{\lambda \rightarrow \infty}{\sum_{n \ge 1}{\frac{\lambda\psi_n^2}{1+\frac{\lambda}{D_1}\lambda_n}}} 
= D_1 \sum_{n \ge 1}{\frac{\psi_n^2}{\lambda_n}},
\end{align*}
we deduce
\begin{align}\label{eq:<T(1),1>}
\langle \tilde{T}(1),1 \rangle = \sum_{n \ge 1}{\frac{\psi_n^2}{\lambda_n}} ,
\end{align}
and this series converges.  Setting $\lambda = 0$ into
Eq. \eqref{eq:<t1>2}, the above expression can be identified to the
mean exit time for surface diffusion phase:
\begin{equation}
\label{eq:t1_surface}
\langle t_1 \rangle_{\lambda = 0} =  \frac{1}{\pi D_1} \sum_{n \ge 1}{\frac{\psi_n^2}{\lambda_n}} = 
\frac{1}{\pi D_1} \langle \tilde{T}(1),1 \rangle = \frac{(\pi - \epsilon)^3}{3\pi D_1} ,
\end{equation}
because $\tilde{T}(1) = \frac{(\pi-\epsilon)^2 - \theta^2}{2}$ for
$0\leq \theta < \pi-\epsilon$, and $0$ otherwise.

Plugging Eq. \eqref{eq:<T(1),1>} into Eq. \eqref{eq:<t1>2}, we obtain
\begin{align}\label{eq:<t1>3}
\langle t_1 \rangle = \frac{1}{\pi}\left(\frac{1}{\lambda}+\frac{1-(1-a)^2}{4D_2}\right)\sum_{n \ge 1}
\frac{\psi_n^2}{\lambda_n \left( \frac{D_1}{\lambda}+\lambda_n \right)}.
\end{align}
This formula generalizes Eq. \eqref{eq:meant1_eps0} to extended
targets.

Similarly, we have
\begin{align*}
\lim_{\lambda \rightarrow \infty}{\sum_{n \ge 1}{\frac{\psi_n^2}{\lambda_n \left( \frac{D_1}{\lambda}+\lambda_n \right)}}} 
= \sum_{n \ge 1}{\frac{\psi_n^2}{\lambda_n^2}} ,
\end{align*}
so that
\begin{align}\label{eq:lim<t1>}
\L = \lim_{\lambda \rightarrow \infty}{ \langle t_1 \rangle}= \frac{1-(1-a)^2}{4\pi D_2}\sum_{n \ge 1}{\frac{\psi_n^2}{\lambda_n^2}} ,
\end{align}
i.e., we got a spectral representation for the time $\L$ corresponding
to pure bulk diffusion.  Note also that the series $S= \sum_{n \ge
1}{\frac{\psi_n^2}{\lambda_n^2}}$ converges.

We may now state the main theorem of this work:
\begin{theorem}
\label{thm:theorem1}
The function $\lambda \mapsto \langle t_1 \rangle$ is eventually
increasing as $\lambda \rightarrow +\infty$.  Moreover, if
\begin{align*}
\left. \frac{d \langle t_1 \rangle}{d \lambda}\right|_{\lambda=0} < 0 ,
\end{align*}
then the function $\langle t_1 \rangle$ passes through a minimum.
\end{theorem}

\begin{proof}
\begin{lemma}
\begin{align*}
\sum_{n \ge 1}{\frac{\psi_n^2}{\lambda_n^3}}=+\infty .
\end{align*}
\end{lemma}
The proof of this lemma will be given in the next section.  We assume
it is true for the rest of the proof of the main theorem.

We rewrite Eq. \eqref{eq:<t1>3} as follows
\begin{align}\label{eq:behavior_<t1>}
\langle t_1 \rangle = \L + \frac{S}{\pi \lambda}-\frac{D_1}{\lambda}\frac{1-(1-a)^2}{4\pi D_2}
\sum_{n \ge 1}{\frac{\psi_n^2}{\lambda_n^2\left(\lambda_n+\frac{D_1}{\lambda}\right)}} - \frac{D_1}{\pi \lambda^2}
\sum_{n \ge 1}{\frac{\psi_n^2}{\lambda^2 \left( \lambda_n + \frac{D_1}{\lambda} \right)}} .
\end{align}
Since
\begin{align*}
\lim_{\lambda \rightarrow \infty}{\sum_{n \ge 1}{\frac{\psi_n^2}{\lambda_n^2\left(\lambda_n+\frac{D_1}{\lambda}\right)}}} 
= \sum_{n \ge 1}{\frac{\psi_n^2}{\lambda_n^3}} = +\infty ,
\end{align*}
we see that
\begin{align*}
\langle t_1 \rangle = \L -\frac{D_1}{\lambda}\frac{1-(1-a)^2}{4\pi D_2}\sum_{n \ge 1}
{\frac{\psi_n^2}{\lambda_n^2\left(\lambda_n+\frac{D_1}{\lambda}\right)}} + O(\lambda^{-1}) ,
\end{align*}
from which the principle part of $\langle t_1 \rangle$ is indeed
non-decreasing as $\lambda \rightarrow \infty$.

To complete the proof, we differentiate Eq. \eqref{eq:behavior_<t1>}
(the exchange of derivative and sum is easily established):
\begin{align*}
\pi \frac{d \langle t_1 \rangle}{d\lambda} = -\frac{S}{\lambda^2} +
\frac{1}{\lambda^2}\frac{D_1\left(1-(1-a)^2\right)}{4D_2}
\sum_{n \ge 1}\frac{\psi_n^2}{\lambda_n \left(\lambda_n+\frac{D_1}{\lambda}\right)^2} + 
\frac{D_1}{\lambda^3} \sum_{n \ge 0}{\frac{\psi_n^2 \left( 2 \lambda_n + D_1/ \lambda \right)}{\lambda_n^2 \left( \lambda_n + D_1/ \lambda \right)^2}} ,
\end{align*}
which becomes positive at large $\lambda$ because
\begin{align*}
\lim_{\lambda \rightarrow \infty}\sum_{n \ge 1}{\frac{\psi_n^2}{\lambda_n}\frac{1}{\left(\lambda_n+\frac{D_1}{\lambda}\right)^2}} 
= \sum_{n \ge 1}{\frac{\psi_n^2}{\lambda_n^3}} = +\infty.
\end{align*}
Returning to formula \eqref{eq:<t1>2}, we have
\begin{align*}
\pi \frac{d \langle t_1 \rangle}{d\lambda} = &\frac{1-(1-a)^2}{4D_1D_2}\left[ \langle\tilde{T}(1),1 \rangle - 
\frac{\lambda}{D_1}\sum_{n \ge 1}{\frac{\psi_n^2}{1+\frac{\lambda \lambda_n}{D_1}}} \right] \\
&+\frac{1}{D_1}\left(1+\frac{\lambda\left(1-(1-a)^2\right)}{4D_2}\right)\left[-\frac{1}{D_1}
\sum_{n \ge 1}{\frac{\psi_n^2}{1+\frac{\lambda \lambda_n}{D_1}}} + \frac{\lambda}{D_1}
\sum_{n \ge 1}{\frac{\psi_n^2 \frac{\lambda_n}{D_1}}{\left(1+\frac{\lambda \lambda_n}{D_1}\right)^2}}\right].
\end{align*}
In particular, one gets
\begin{align*}
\pi \left.\frac{d \langle t_1 \rangle}{d\lambda}\right|_{\lambda=0} &= \frac{1-(1-a)^2}{4D_1D_2} \langle\tilde{T}(1),1 \rangle
+ \frac{1}{D_1^2}\left[ - \sum_{n \ge 1}{\psi_n^2} \right]\\
&=\frac{1-(1-a)^2}{4D_1D_2} \langle\tilde{T}(1),1 \rangle-\frac{1}{D_1^2}\| V\tilde{T}(1) \|^2 ,
\end{align*}
which becomes negative when
\begin{align}
\| V\tilde{T}(1) \|^2 & > \frac{D_1}{4D_2}\left(1-(1-a)^2\right)\langle \tilde{T}(1),1 \rangle \label{eq:ratio_of_D1/D2}\\
\Leftrightarrow \| V\tilde{T}V(\psi_0) \|^2 & > \frac{D_1}{4D_2}\left(1-(1-a)^2\right)\langle V\tilde{T}V(\psi_0),\psi_0 \rangle .
\end{align}
Note that
\begin{align}
\psi & = V\tilde{T}(1) = V\left(\frac{(\pi-\epsilon)^2-\theta^2}{2}\right) = V\left(\frac{2}{\pi}\sum\limits_{n\geq 1}^\infty \cos n\theta
~ \langle \frac{(\pi-\epsilon)^2-\theta^2}{2} , \cos n\theta \rangle \right)  \nonumber\\
& = \frac{2}{\pi}\sum\limits_{n\geq 1}^\infty \cos n\theta  ~ \sqrt{1-(1-a)^n} ~
\langle \frac{(\pi-\epsilon)^2-\theta^2}{2} , \cos n\theta \rangle \nonumber \\
& = \frac{2}{\pi}\sum\limits_{n\geq 1}^\infty \cos n\theta  ~ \sqrt{1-(1-a)^n} ~ 
(-1)^{n-1} \frac{(\pi-\epsilon)\cos n\epsilon + \frac{\sin n\epsilon}{n}}{n^2} .
 \nonumber 
\end{align}

\begin{remark}
The inequality \eqref{eq:ratio_of_D1/D2} determines the critical value
for the pure bulk diffusion coefficient $D_{2,\rm crit}$ above which pure bulk
excursions are beneficial.  The existence of the optimal value
$\lambda$ (that minimizes the function $\langle t_1 \rangle$) depends
on this ratio.  If $D_2 > D_{2,\rm crit}$, with
\begin{align*}
D_{2,\rm crit} & = D_1 \frac{(1-(1-a)^2) \langle \tilde{T}(1),1 \rangle}{4 \Vert V\tilde{T}(1) \Vert^2} \\
& = D_1 \frac{\pi(\pi-\epsilon)^3(1-(1-a)^2)}{24} \left(
\sum_{n \ge 1} \frac{1-(1-a)^n}{n^4}\left[ (\pi-\epsilon)\cos n\epsilon+\frac{\sin n\epsilon}{n} \right]^2 \right)^{-1} , \\
\end{align*}
then $\langle t_1 \rangle$ starts first to decrease with $\lambda$,
passes through a minimum then monotonously increases.  We retrieved
the optimality condition first reported in \cite{Benichou10}.
\end{remark}
\end{proof}

\begin{proof}[Proof of the main lemma]
First we define $\tilde{e}_n = \frac{1}{\lambda_n}\tilde{T}Ve_n$ so
that $V\tilde{e}_n=e_n$.

Let $u$ be such that $\psi=V\tilde{T}V(u)$; $u$ must be of the form
$\psi_0+u^{\perp}$, where $u^{\perp} \in \ker (V\tilde{T}V)$.  Let
\begin{align*}
u_n &= \langle u,e_n \rangle = \frac{1}{\lambda_n} \langle u, V\tilde{T}V e_n \rangle 
= \frac{1}{\lambda_n} \langle V\tilde{T}V u, e_n \rangle = \frac{\psi_n}{\lambda_n} .
\end{align*}

\begin{remark}
This computation gives another proof that $\sum_{n \ge
1}{\frac{\psi_n^2}{\lambda_n^2}} <\infty$ and thus also that $\sum_{n
\ge 1}{\frac{\psi_n^2}{\lambda_n}} <\infty$.
\end{remark}

Now, on $[0,\pi-\epsilon]$, $-\lambda_n \tilde{e}''_n = V e_n$ because
$\tilde{e}_n = \frac{1}{\lambda_n} \tilde{T}V e_n$.  Let $m,n$ be
integers, and we establish:
\begin{align*}
\lambda_n \langle \tilde{e}'_n, \tilde{e}'_m \rangle = -\lambda_n \langle \tilde{e}''_n, \tilde{e}_m \rangle = \langle Ve_n, \tilde{e}_m \rangle 
= \langle e_n, V\tilde{e}_m \rangle = \langle e_n, e_m \rangle =\delta_{m,n}.
\end{align*}
Setting $\epsilon_n = \sqrt{\lambda_n}\tilde{e}_n$, we get that
$(\epsilon'_n)$ is an orthonormal system of $L^2\left( [0,\pi]
\right)$.

Next assume that $\sum_{n \ge 1}{\frac{\psi_n^2}{\lambda_n^3}}
<\infty$: then the above computation shows that $\sum_{n \ge
1}{\frac{\psi_n}{\lambda_n}\tilde{e}_n}$ is a function of the Sobolev
space $H^1\left( [0,\pi] \right)$.

Before we continue let us observe that the operator $I-V$ is
regularizing: $\forall f \in L^2\left( [0,\pi] \right)$, $Vf=f+g$
where $g \in C^{\infty}\left( [0,\pi] \right)$.

This implies that
\begin{align*}
V\left( \sum_{n \ge 1}{\frac{\psi_n}{\lambda_n}\tilde{e}_n} \right) = \sum_{n \ge 1}{\frac{\psi_n}{\lambda_n}\tilde{e}_n}+g , \quad g\in C^{\infty}.
\end{align*}
On the other hand
\begin{align*}
V\left( \sum_{n \ge 1}{\frac{\psi_n}{\lambda_n}\tilde{e}_n} \right) = \sum_{n \ge 1}{\frac{\psi_n}{\lambda_n}e_n} ,
\end{align*}
and thus $u_0=\sum_{n \ge 1}{\frac{\psi_n}{\lambda_n}e_n} \in H^1$.

But $u_0$ minimizes $\|v\|_2^2$ on the set of $v$ such that
$\int_{0}^{\pi}{(Vv-1)^2}=0$,
$\int_{\pi-\epsilon}^{\pi}{Vv}=-(\pi-\epsilon)$.

By the theory of constrained extrema, $v$ must be of the form $\lambda
V 1_{[\pi-\epsilon,\pi[} = \lambda 1_{[\pi-\epsilon,\pi[}+g$ with $g
\in C^{\infty}$. But such a function cannot be in $H^1$.

We have thus proven the main lemma.
\end{proof}

\section{Asymptotic behavior}

Here we analyze the asymptotic behavior of $\langle t_1
\rangle$ as $\lambda \to \infty$.  

For point-like target ($\epsilon = 0$), Eqs. \eqref{eq:eigen_eps0} and
\eqref{eq:psi_n} imply for small $a$ the existence of two
distinct asymptotic behaviors:
\begin{align}
\label{eq:lambda_asympt_eps0a}
\lambda_n & \simeq \frac{a}{n} ,   \quad \psi_n^2 \simeq \frac{2\pi a}{n^3}  \qquad (n \ll 1/a),  \\
\label{eq:lambda_asympt_eps0}
\lambda_n & \simeq \frac{1}{n^2} , \quad \psi_n^2 \simeq \frac{2\pi}{n^4} \qquad (n \gg 1/a) .
\end{align}

\begin{figure}[H]
\begin{center}
	\includegraphics[width=0.49 \textwidth]{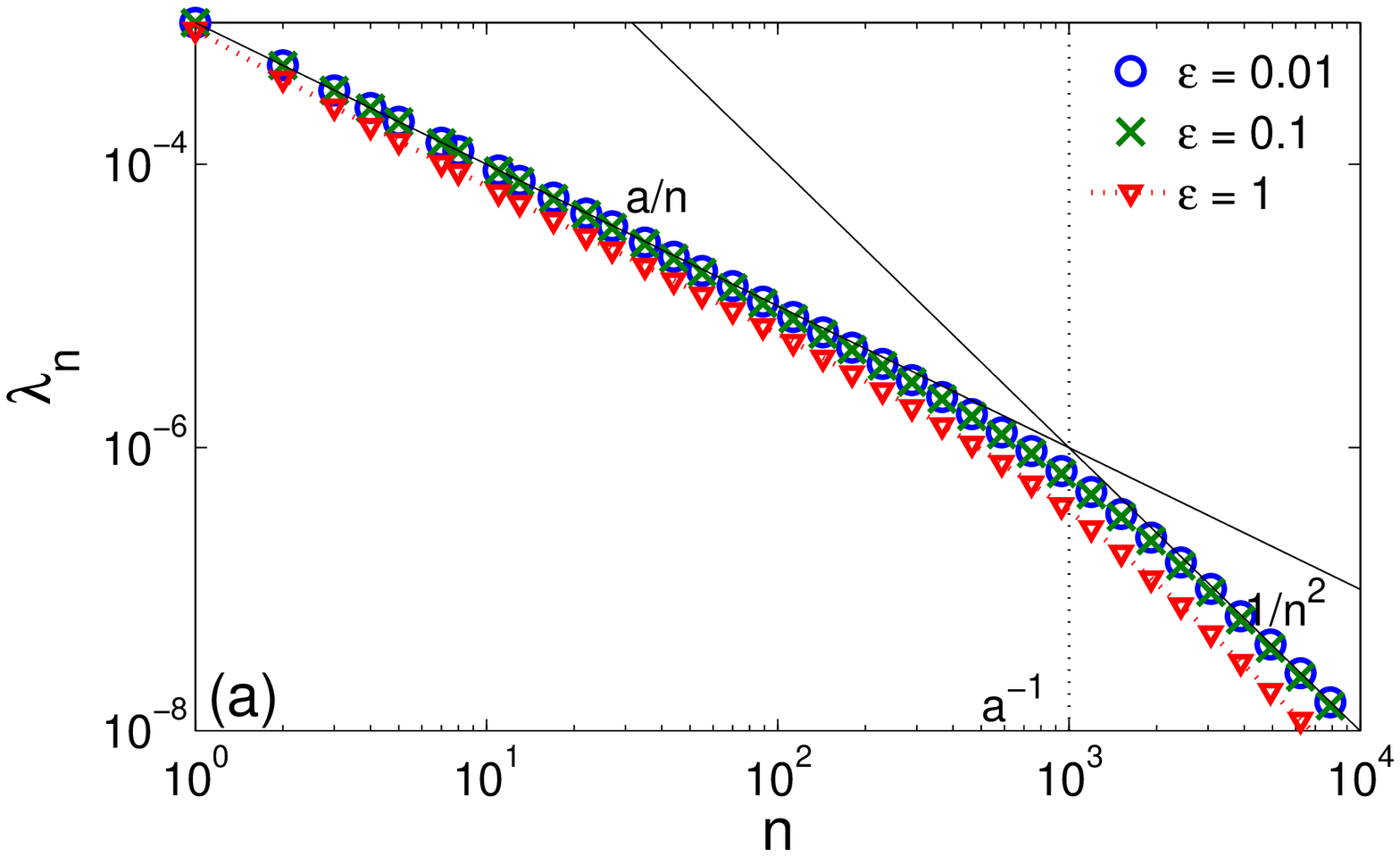}
	\includegraphics[width=0.49 \textwidth]{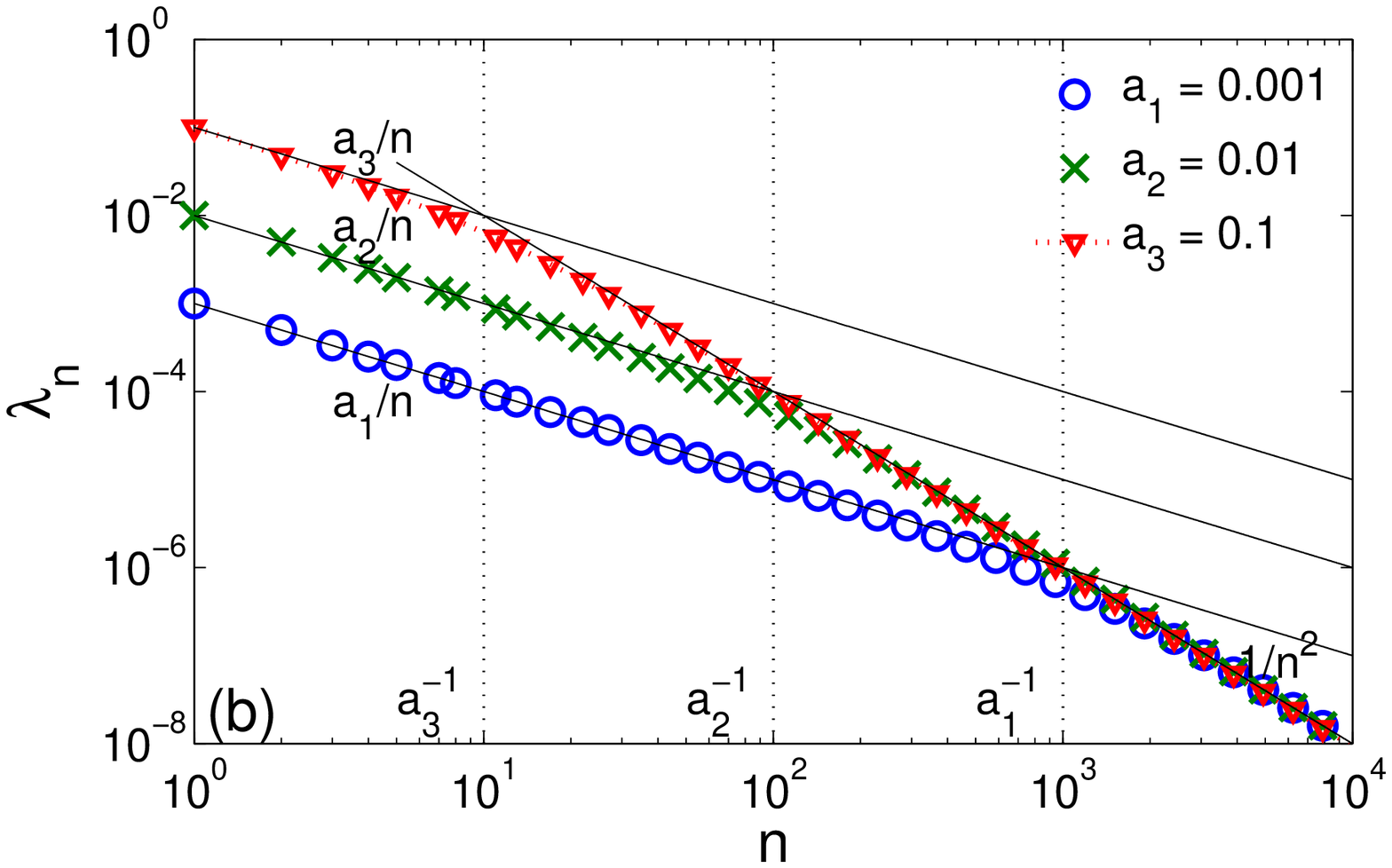}
\end{center}
\caption{
Eignvalues $\lambda_n$ of the operator $V\tilde{T}V$ for ({\bf a}) $a
= 0.001$ and three values $\epsilon$: $0.01$ (circles), $0.1$
(crosses), and $1$ (triangles); and ({\bf b}) for $\epsilon = 0.01$
and three values of $a$: $0.001$ (circles), $0.01$ (crosses), and
$0.1$ (triangles).  Solid lines show the asymptotic relations $a/n$
and $1/n^2$, while vertical dotted lines indicate the separation $1/a$
between these asymptotic regimes.  The coefficient $A_{\epsilon}$ in
front of $n^{-2}$ relation is close to $1$ for all small targets,
except for $\epsilon = 1$, see Eq. \eqref{eq:Aepsilon}.}
\label{fig:lambda}
\end{figure}

In this section, we analyze the case of extended targets ($\epsilon >
0$) by numerical computations of the eigenvalues $\lambda_n$ and the
spectral weights $\psi_n^2$ (see Appendix \ref{sec:numerics} for
computational details).  As for the point-like target, Figure
\ref{fig:lambda} allows one to distinguish two regimes for
$\lambda_n$, for small and large $n$:
\begin{equation}
\label{eq:asympt1}
\lambda_n \simeq \begin{cases} \tilde{A}_{a,\epsilon}~ n^{-1}    \quad (n \ll 1/a) , \cr
			      A_{\epsilon}~ n^{-2}  \hskip 6mm (n \gg 1/a) ,  \end{cases}
\end{equation}
where $\tilde{A}_{a,\epsilon}$ and $A_{\epsilon}$ are two constants.
In Appendix \ref{sec:asympt}, we prove the second (large $n$)
asymptotic relation and show that the constant $A_{\epsilon}$ does not
depend on $a$.  In turn, the transition between two asymptotic regimes
is determined by $1/a$ (and is independent of $\epsilon$).  Note also
that $A_\epsilon \to 1$ as $\epsilon\to 0$ according to
Eq. \eqref{eq:lambda_asympt_eps0}.  One can see that these asymptotic
relations accurately approximate the eigenvalues $\lambda_n$.  The
behavior of $A_\epsilon$ is shown on Fig. \ref{fig:Aeps}a.  As
expected, it does not depend on $a$.  These numerical results suggest
the following conjectural expression:
\begin{equation}
\label{eq:Aepsilon}
A_\epsilon = ( 1 - \epsilon/\pi)^2 ,
\end{equation}
which accurately reproduces $A_\epsilon$ on the whole range of
$\epsilon$ from $0$ to $\pi$.  According to
Eq. \eqref{eq:lambda_asympt_eps0a}, the coefficient
$\tilde{A}_{a,\epsilon}$ is equal to $a$ when $\epsilon = 0$.  We plot
therefore $\tilde{A}_{a,\epsilon}/a$ on Fig. \ref{fig:Aeps}b, where
this ratio approaches $1$ as $\epsilon\to 0$, and $0$ as $\epsilon \to
\pi$.  Moreover, this ratio weakly depends on $a$ (curves for $a =
0.001$ and $a = 0.1$ almost coincide).

\begin{figure}
\begin{center}
	\includegraphics[width=0.49 \textwidth]{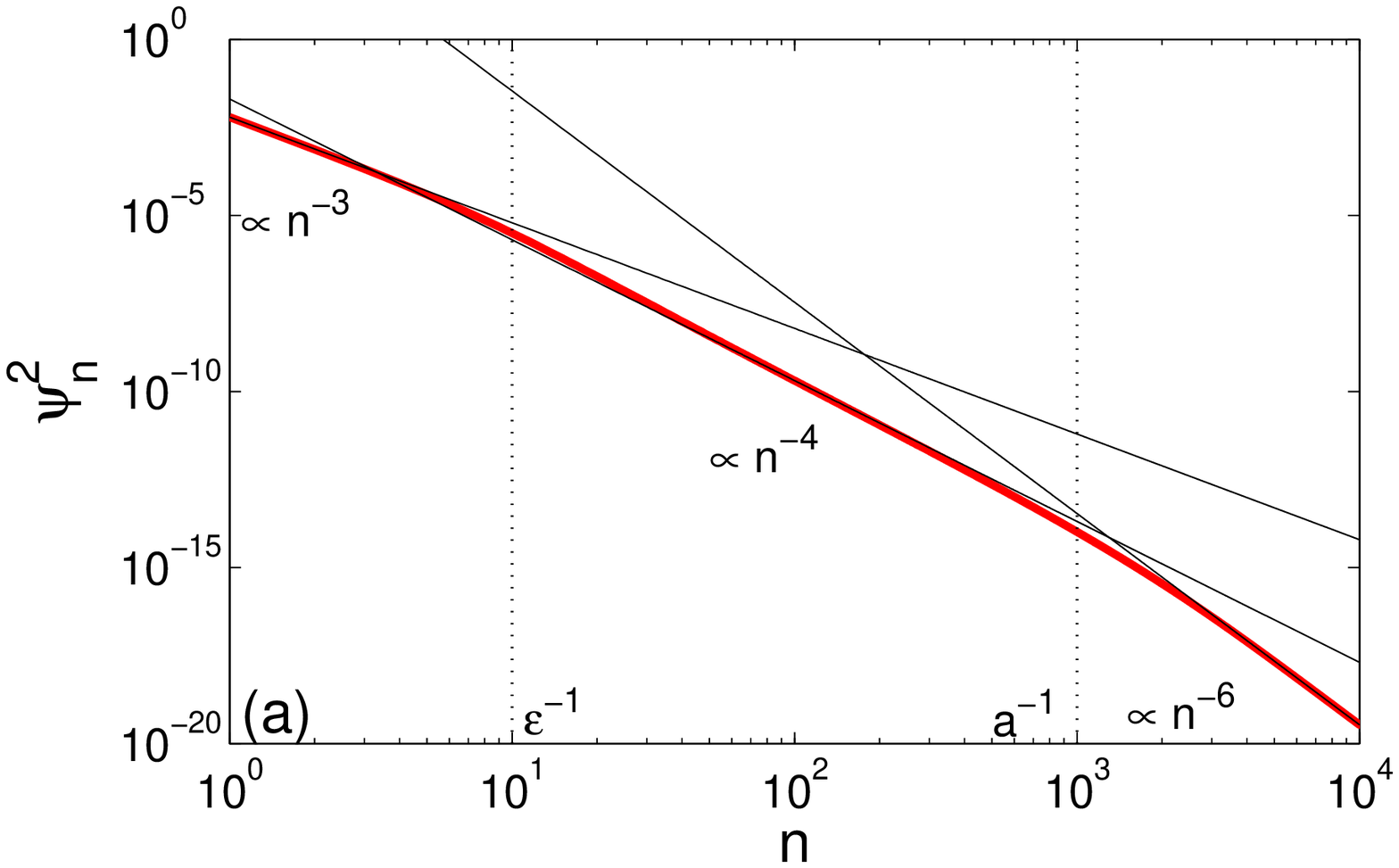}
	\includegraphics[width=0.49 \textwidth]{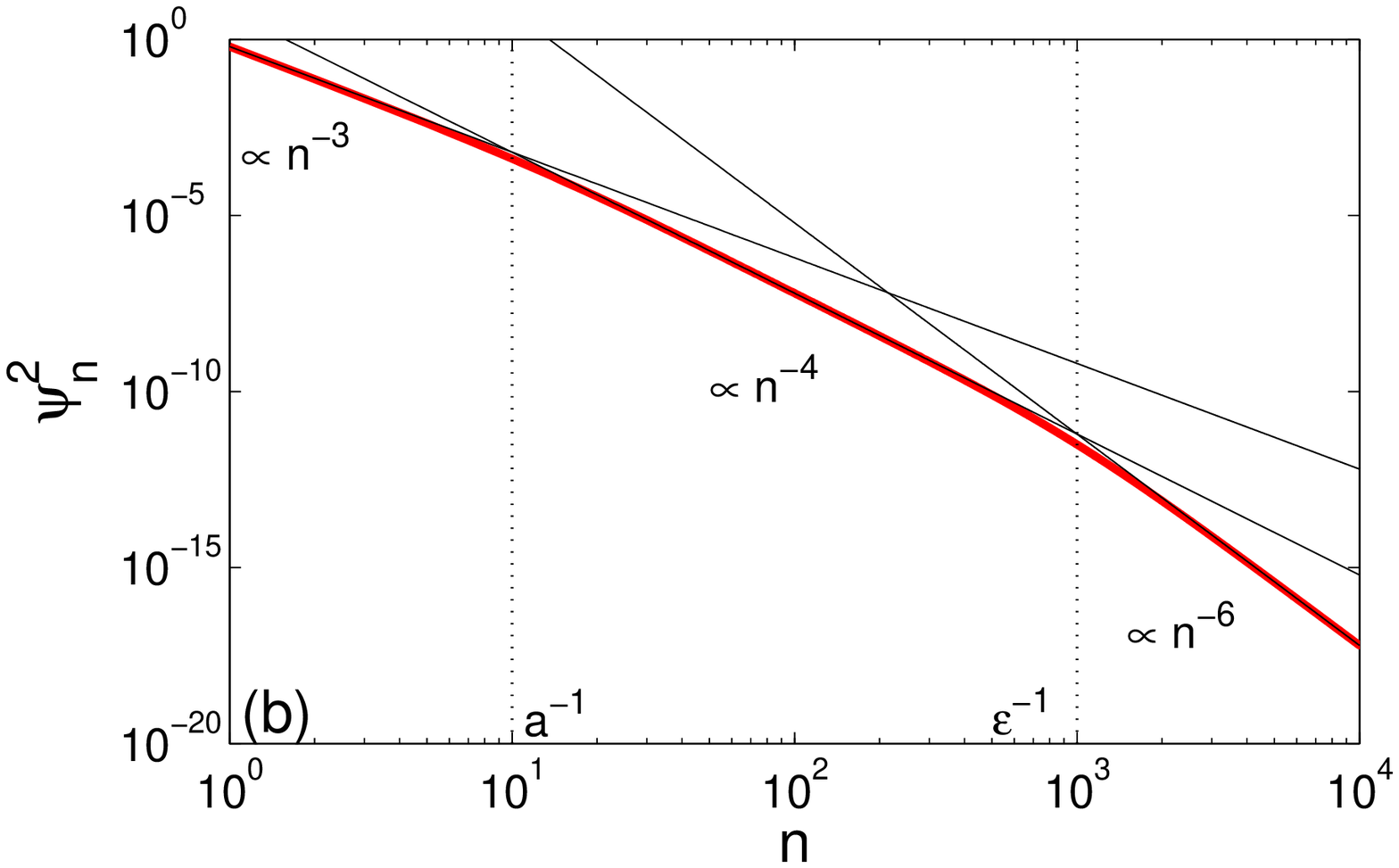}
	\includegraphics[width=0.49 \textwidth]{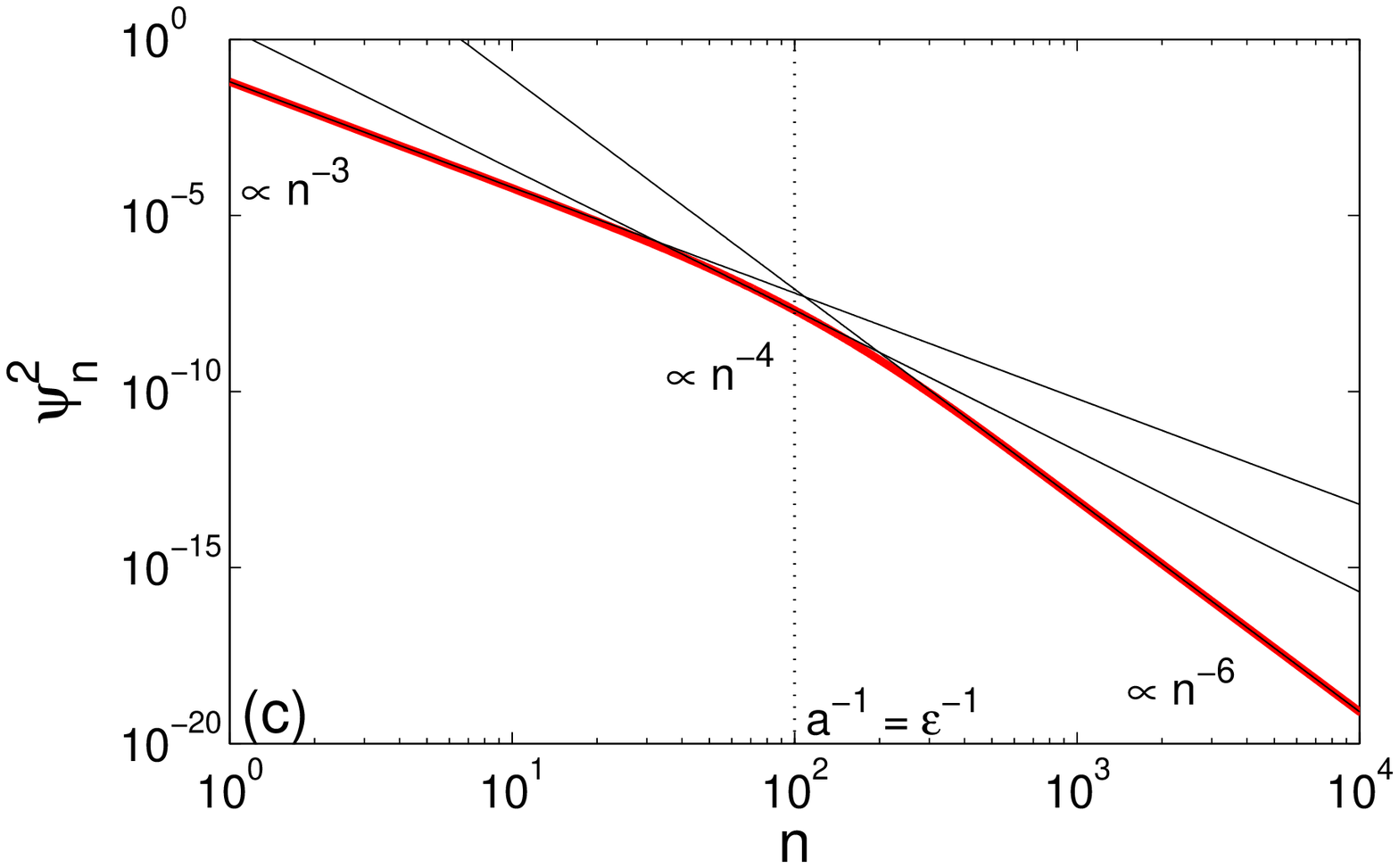}
\end{center}
\caption{
Spectral weights $\psi_n^2$ (shown by red solid line) for ({\bf a}) $a
= 0.001$ and $\epsilon = 0.1$; ({\bf b}) $a = 0.1$ and $\epsilon =
0.001$; and ({\bf c}) $a = \epsilon = 0.01$.  In the first two plots,
three asymptotic regimes can be distinguished according to
Eq. \eqref{eq:asympt2}, while the intermediate regime disappears in
the last plot.}
\label{fig:psi}
\end{figure}

Figure \ref{fig:psi} shows that the asymptotic behavior of the
spectral weights $\psi_n^2$ is more complicated.  One can distinguish
three asymptotic regimes:
\begin{equation}
\label{eq:asympt2}
\psi_n^2 \simeq \begin{cases}   \tilde{B}_{a,\epsilon}~ n^{-3}    \quad (n \ll \min\{1/a, 1/\epsilon\}) , \cr
				\tilde{B}'_{a,\epsilon}~ n^{-4}    \quad (\min\{1/a, 1/\epsilon\} \ll n \ll \max\{1/a,1/\epsilon\}) , \cr
			        B_{a,\epsilon}~ n^{-6}  \quad (n \gg \max\{1/a,1/\epsilon\}) ,  \end{cases}
\end{equation}
In order to observe all three regimes, one needs $1 \ll \min\{1/a,
1/\epsilon\} \ll \max\{1/a, 1/\epsilon\}$, i.e., either $a \ll
\epsilon \ll 1$, or $\epsilon \ll a \ll 1$.  For instance, if $a$ or
$\epsilon$ is not small enough, the first regime with $n^{-3}$ may not
be well established (Fig. \ref{fig:psi}a, b).  If $a \sim \epsilon$,
the intermediate regime disappears, as illustrated on
Fig. \ref{fig:psi}c.  Finally, when $\epsilon \to 0$,
$\max\{1/a,1/\epsilon\} \to \infty$, the third regime disappears, and
one retrieves two regimes for point-like targets.  

The behavior of the coefficients $B_{a,\epsilon}$ and
$\tilde{B}_{a,\epsilon}$ is shown on Fig. \ref{fig:Aeps}c,d.  As
expected from Eq. \eqref{eq:lambda_asympt_eps0a},
$\tilde{B}_{a,\epsilon}/(2\pi a)$ approaches $1$ as $\epsilon\to 0$
(point-like target).  Moreover, such normalized coefficient weakly
depends on $a$ (at least for small $a$).  The behavior of
$B_{a,\epsilon}$ is more complicated.  Given that $\L$ should converge
to the mean exit time for pure bulk diffusion, $\langle t_1\rangle_b$, as
$a\to 0$, one gets
\begin{equation}
\sum\limits_{n\geq 1} \frac{\psi_n^2}{\lambda_n^2} \simeq \frac{2 \pi D_2 \langle t_1\rangle_b}{a} .
\end{equation}
Since $\lambda_n \simeq A_\epsilon n^{-2}$ is independent of $a$, one
concludes that $B_{a,\epsilon} \sim 1/a$ as $a\to 0$.  For this
reason, we plot $a B_{a,\epsilon}$ on Fig. \ref{fig:Aeps}c.  For large
$\epsilon$, two curves for $a = 0.001$ and $a = 0.1$ do coincide, as
expected.  However, strong deviations emerge at small $\epsilon$.  In
fact, one needs to consider much smaller $a$ to get coinciding curves
over the whole considered range of $\epsilon$.  We conclude that the
reflection distance $a$ plays an important role, especially for small
targets.

Although the above asymptotic regimes for $\lambda_n$ and $\psi_n^2$
remain conjectural, we will investigate their consequences for the
asymptotic behavior of the mean exit time $\langle t_1 \rangle$.
Using the asymptotic relations for large $n$, we get
\begin{align}
\sum_{n\geq 1}{\frac{\psi_n^2}{\lambda_n^2\left(\lambda_n+\frac{D_1}{\lambda}\right)}} 
&\sim \int_1^{\infty}{\frac{B_{a,\epsilon} x^{-6}}{A_{\epsilon}^2 x^{-4} \left(A_{\epsilon} x^{-2}+\frac{D_1}{\lambda}\right)}dx} \nonumber \\
&= \frac{B_{a,\epsilon}}{A_{\epsilon}^3} \int_{1}^{\infty}{\frac{dx}{1+ x^2\frac{D_1}{\lambda A_{\epsilon}}}}  
 \sim \frac{B_{a,\epsilon}}{A_{\epsilon}^{5/2}}\frac{\pi}{2}\sqrt{\frac{\lambda}{D_1}}  .  \nonumber
\end{align}
Consequently, we get 
\begin{align} \label{eq:behavior_of_t1,epsilon>0}
\langle t_1 \rangle = \L - \frac{C_1}{\sqrt{\lambda}} + O\left( \frac{1}{\lambda} \right) ,
\end{align}
where 
\begin{equation}  \label{eq:C1}
C_1 = C_{a,\epsilon} ~ \frac{\sqrt{D_1}}{D_2} , \qquad C_{a,\epsilon} = (1-(1-a)^2) \frac{ B_{a,\epsilon} }{8A_{\epsilon}^{5/2}}.
\end{equation}

\begin{figure}
\begin{center}
	\includegraphics[width=0.49 \textwidth]{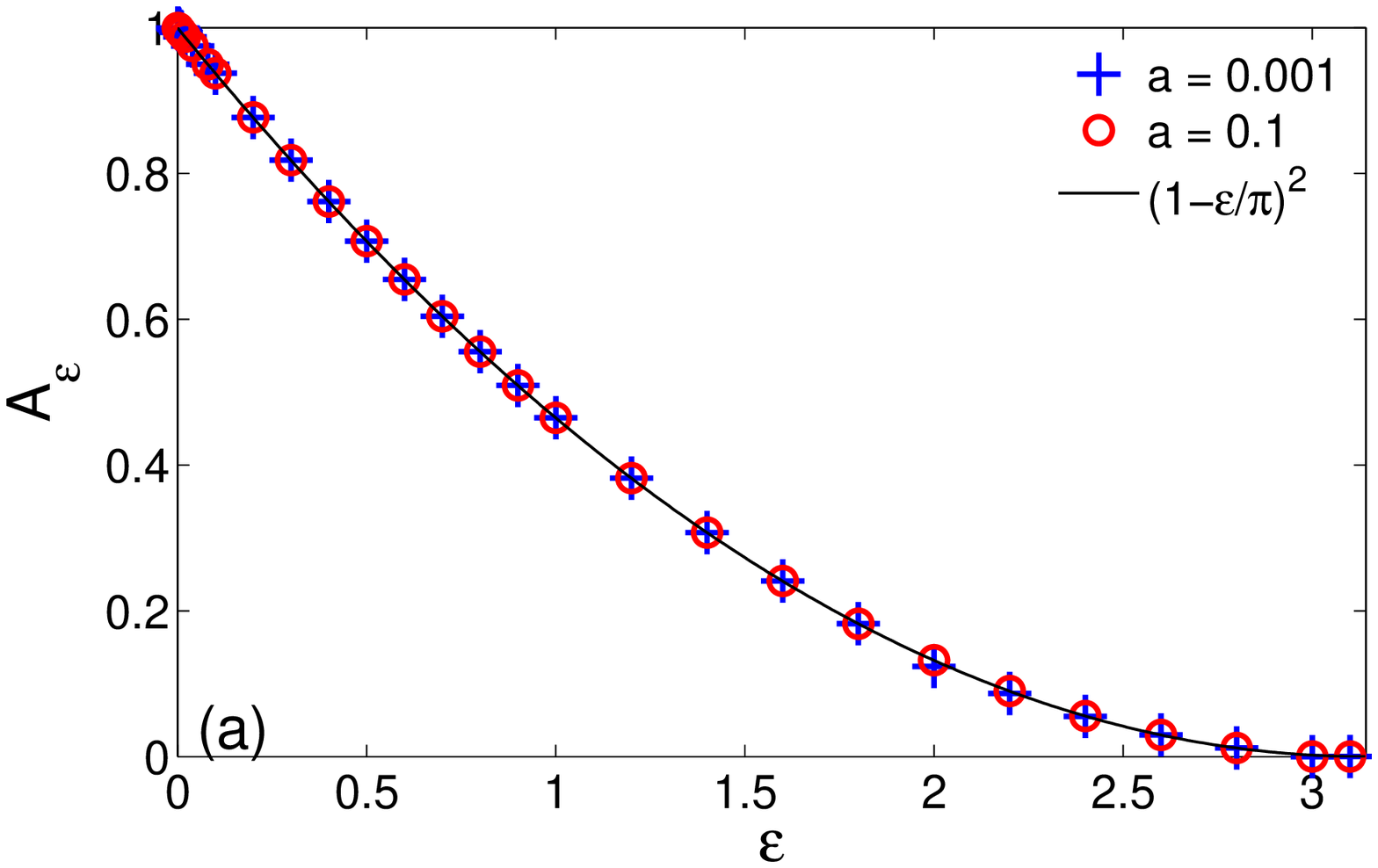} 
	\includegraphics[width=0.49 \textwidth]{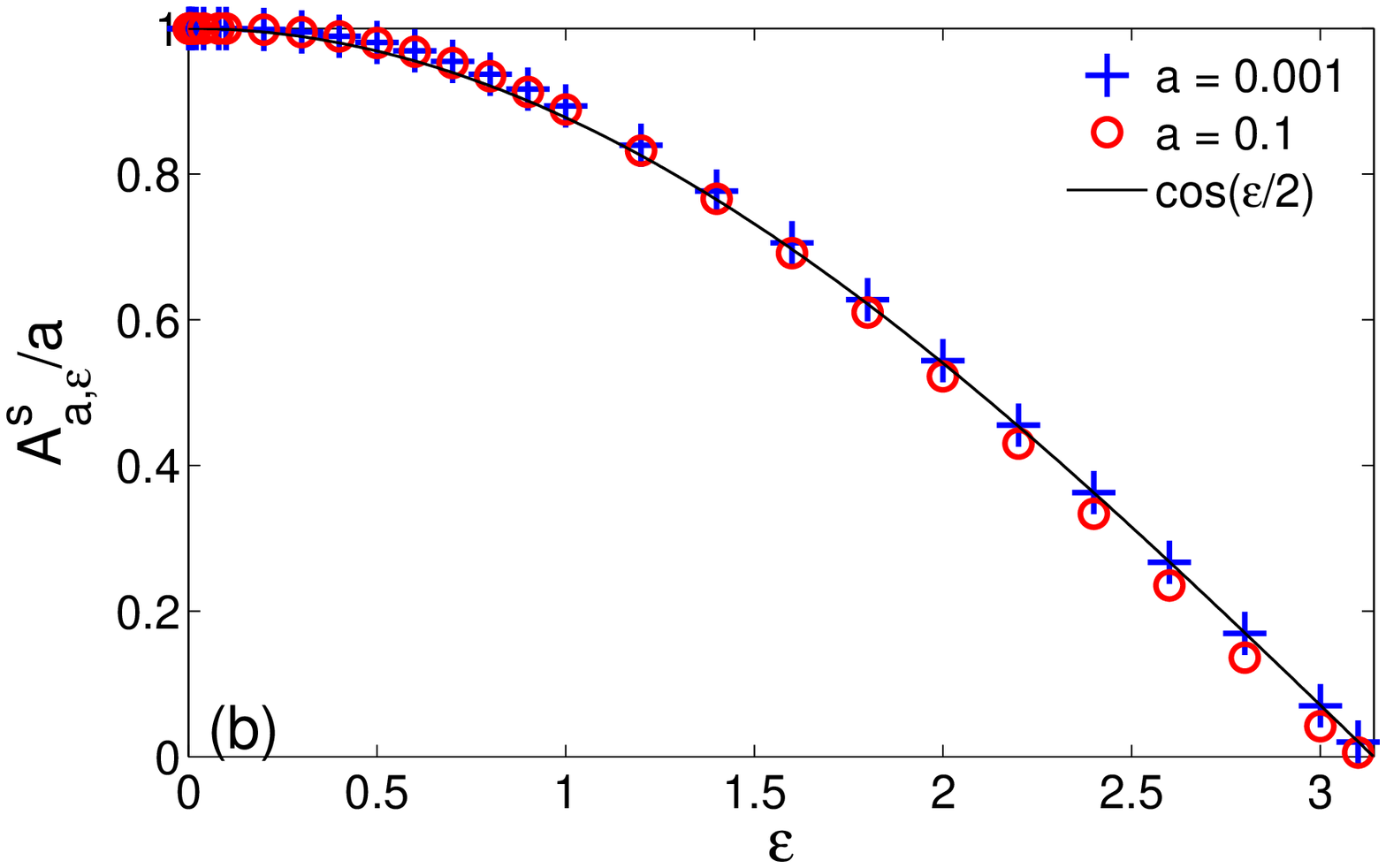} 
	\includegraphics[width=0.49 \textwidth]{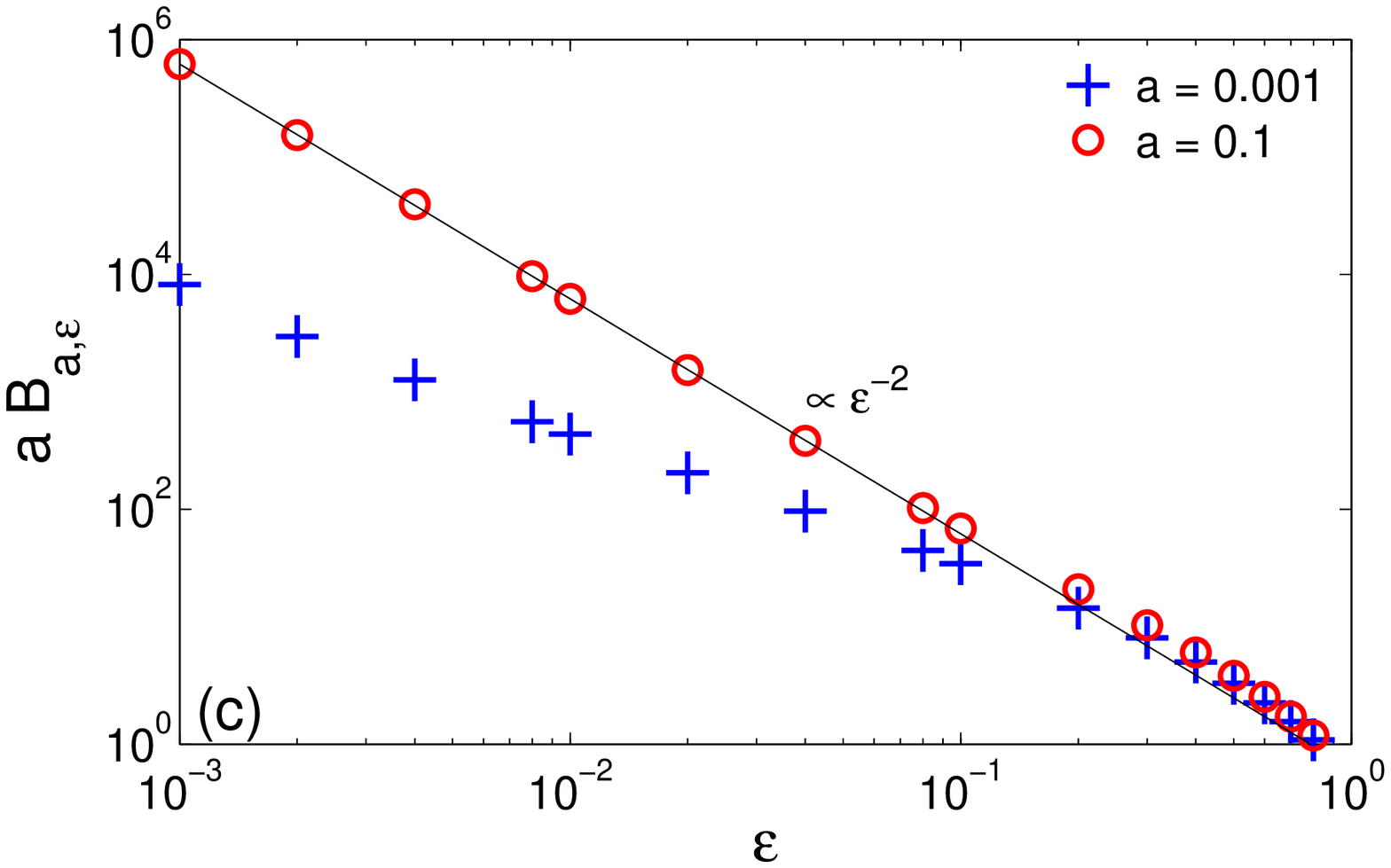} 
	\includegraphics[width=0.49 \textwidth]{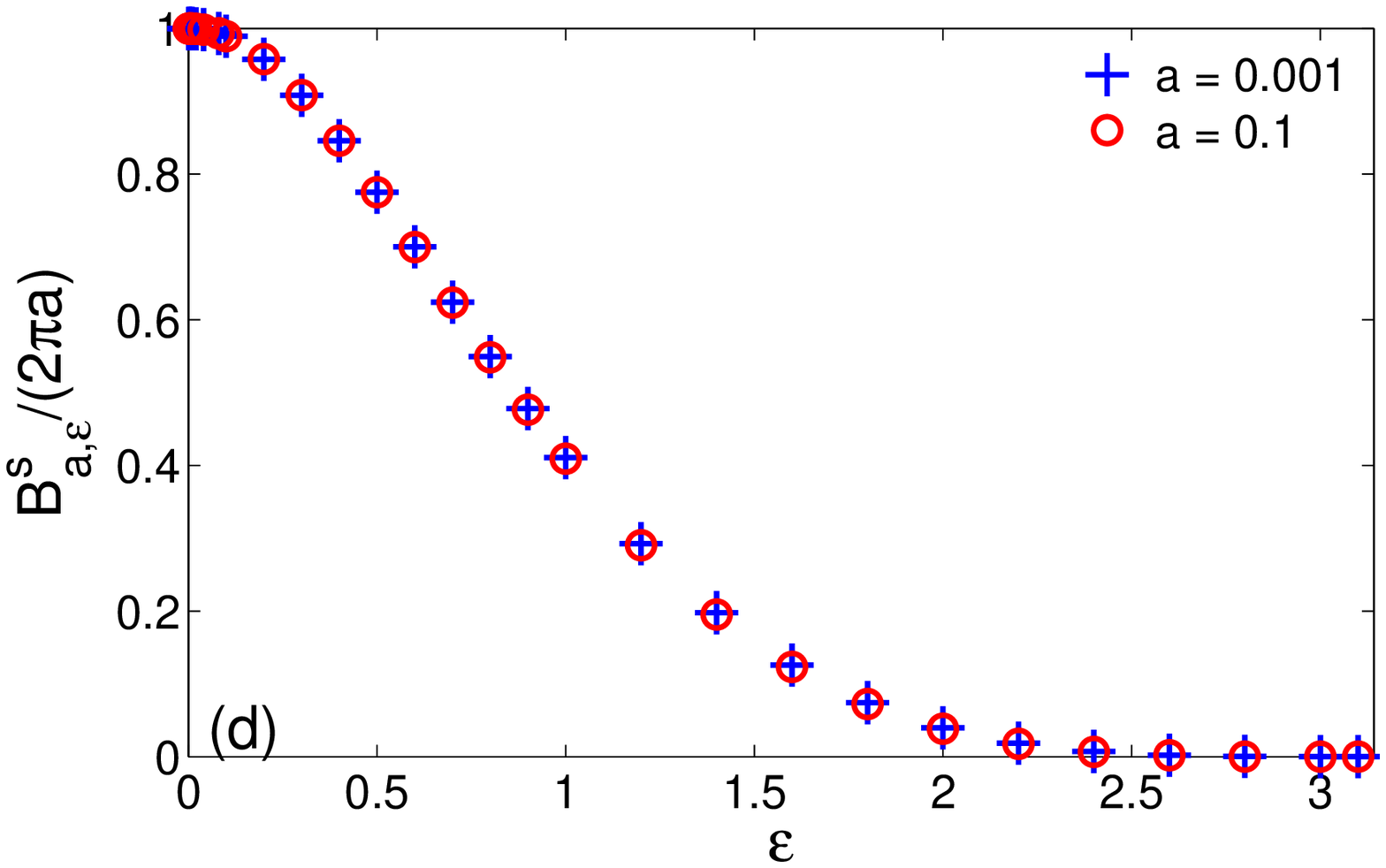} 
\end{center}
\caption{
Coefficients $A_\epsilon$, $\tilde{A}_{a,\epsilon}/a$, $a
B_{a,\epsilon}$, and $\tilde{B}_{a,\epsilon}/a$ from
Eqs. \eqref{eq:asympt1}, \eqref{eq:asympt2} versus $\epsilon$.  Two
curves for $a = 0.001$ and $a=0.1$ coincide that illustrates the
independence of $A_\epsilon$ and $\tilde{A}_\epsilon/a$ of $a$. }
\label{fig:Aeps}
\end{figure}

For an accurate numerical computation of $\langle t_1 \rangle$, we
consider the behavior of partial sums $f(N) = \sum_{n=1}^{N}
\frac{\psi_n^2}{\lambda_n\left(\frac{D_1}{\lambda}+\lambda_n
\right)}$ (note that $\langle t_1 \rangle$ is obtained in the limit $N \to
\infty$ according to Eq. \eqref{eq:<t1>3}).  We checked that $f(N) \sim
f(\infty)+\frac{c}{N}$ for large $N$ (for a fixed $\lambda$).  In
practice, we used the fourth order polynomial fit of $f(N)$ versus
$1/N$ for $N$ from 1000 to 20000 to extrapolate the value $f(\infty)$.

Figure \ref{fig:t1} shows the mean exit time $\langle t_1\rangle$ as a
function of $\lambda$ for a small target ($\epsilon = 0.01$) and two
values of $a$: $0.01$ and $0.001$.  In both cases, the mean exit time
passes through a minimum at some intermediate desorption rate
$\lambda_c$ and then approaches the maximum as $\lambda\to\infty$.
One can clearly see that the optimal value $\lambda_c$, as well as the
height of the maximum at $\lambda \to
\infty$, depend on $a$.  Although both considered values $a = 0.001$
and $a = 0.01$ are small, the limiting mean exit time $\L$ changes
significantly.  The asymptotic relation
\eqref{eq:behavior_of_t1,epsilon>0} (shown by thin solid lines)
accurately captures the limiting behavior.

\begin{figure}
\begin{center}
	\includegraphics[width=0.49 \textwidth]{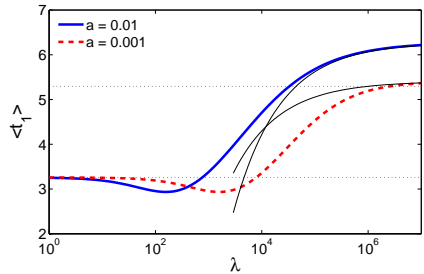}
\end{center}
\caption{
The mean exit time $\langle t_1\rangle$ as a function of $\lambda$ for
$\epsilon = 0.01$, $D_1 = D_2 = 1$, and two values of $a$: $0.01$
(solid line) and $0.001$ (dashed line).  Two horizontal lines indicate
the mean exit times for surface diffusion $\langle t_1\rangle_{\lambda
= 0} \approx 3.2586$ and for pure bulk diffusion $\langle
t_1\rangle_{\lambda = \infty} \approx 5.2929$ from
Eqs. \eqref{eq:t1_surface}, \eqref{eq:t1_bulk}.  Thin lines show the
asymptotic behavior \eqref{eq:behavior_of_t1,epsilon>0} where $\L$ and
$C_1$ are computed from Eqs. \eqref{eq:lim<t1>}, \eqref{eq:C1}. }
\label{fig:t1}
\end{figure}

When $\epsilon \to 0$, we use Eqs. \eqref{eq:asympt1} and
\eqref{eq:asympt2} to get
\begin{align}
\L &= \frac{1-(1-a)^2}{4\pi D_2}\sum_{n \ge 1}{\frac{\psi_n^2}{\lambda_n^2}} 
= \frac{1-(1-a)^2}{4\pi D_2}\left( \sum_{n=1}^{1/\epsilon} \frac{\tilde{B}_{a,\epsilon} n^{-3}}{\tilde{A}_{a,\epsilon}^2 n^{-2}}
+  O(\epsilon)  \right) \nonumber\\
\label{eq:T_log}
&= \frac{1-(1-a)^2}{4\pi D_2}\left( \frac{\tilde{B}_{a,\epsilon}}{\tilde{A}_{a,\epsilon}^2}\ln\frac{1}{\epsilon}+ O(\epsilon) \right) .
\end{align}
This logarithmic divergence is similar to the result from
Ref. \cite{Singer06a,Singer06b,Singer06c} which describes the mean
exit time for non-intermittent pure bulk diffusion (2D Brownian motion) in
the narrow escape limit ($\epsilon\to 0$).  This case formally
corresponds to the double limit $\lambda\to\infty$ and $a\to 0$.
Interestingly, the double limit can be taken separately: as
$\lambda\to\infty$, the limiting value $\L$ exists for any finite $a$.
In this way, we extended the classical escape problem for pure bulk
diffusion by including reflections at a finite distance $a$.  In
Appendix \ref{sec:bulk}, we provide the exact formula for the mean
exit time for pure bulk diffusion $\langle t_1\rangle_b$.  In the narrow
escape limit $\epsilon\to 0$, one gets $\langle t_1\rangle_b \approx
\frac{\ln(2/\epsilon)}{D_2} \bigl(1 + O(\epsilon)\bigr)$.

Figure \ref{fig:Tinf} shows the dependence of the limiting mean exit
time $\L$ on $a$ and $\epsilon$.  For fixed target size $\epsilon$,
$\L$ approaches the mean exit time for pure bulk diffusion $\langle
t_1\rangle_b$ as $a\to 0$ (Fig. \ref{fig:Tinf}a).  Note that a simple
formula for $\L$ in the special case $a = 1$ is derived in Appendix
\ref{sec:transportation}.  For fixed reflection distance $a$, $\L$
exhibits the logarithmic dependence on $\epsilon$ for $\epsilon
\gtrsim a$, in agreement with Eq. \eqref{eq:T_log}.  In the opposite
case $\epsilon \lesssim a$, the reflection distance $a$ alters this
behavior.

\begin{figure}
\begin{center}
	\includegraphics[width=0.49 \textwidth]{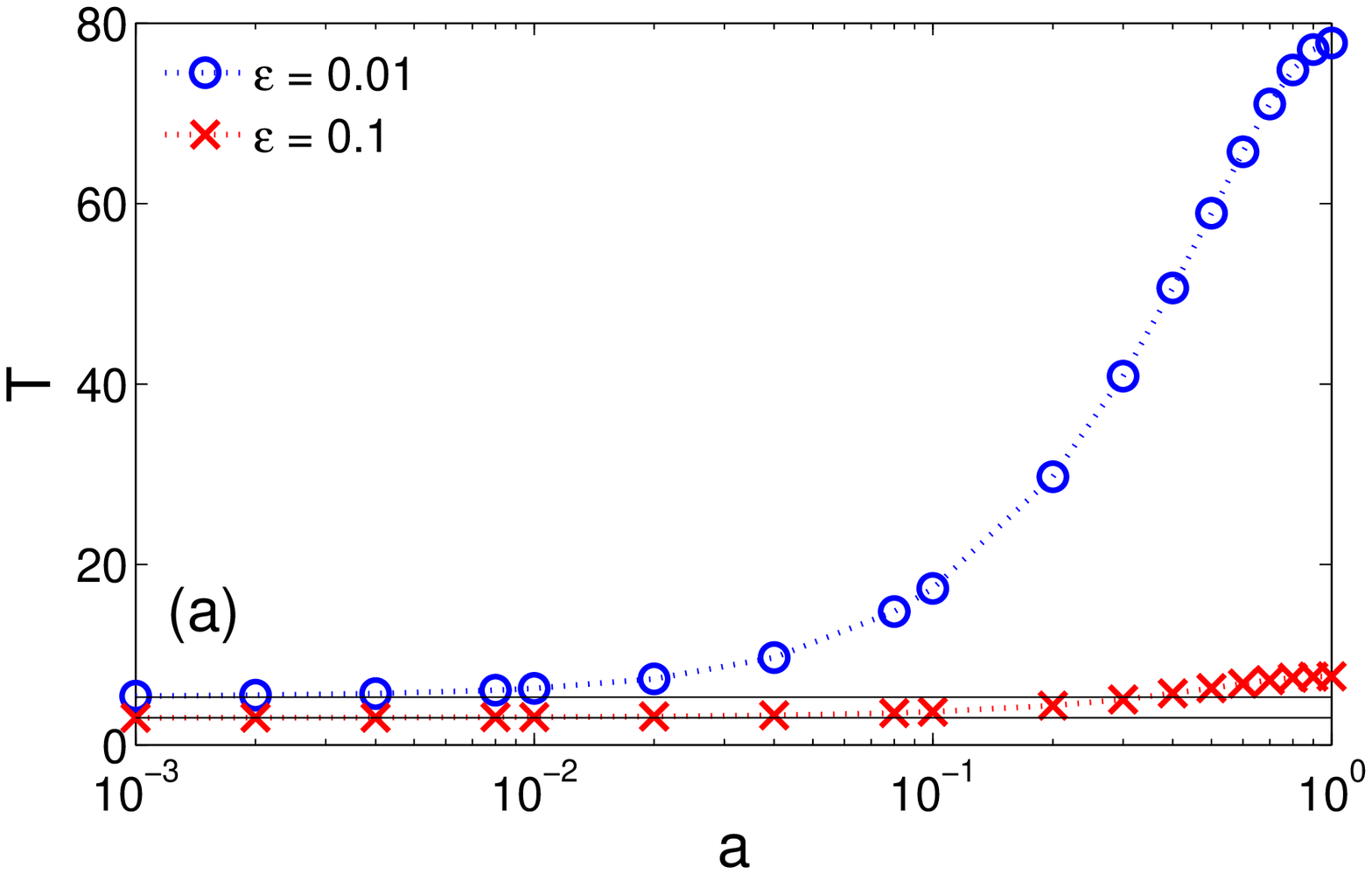}
	\includegraphics[width=0.49 \textwidth]{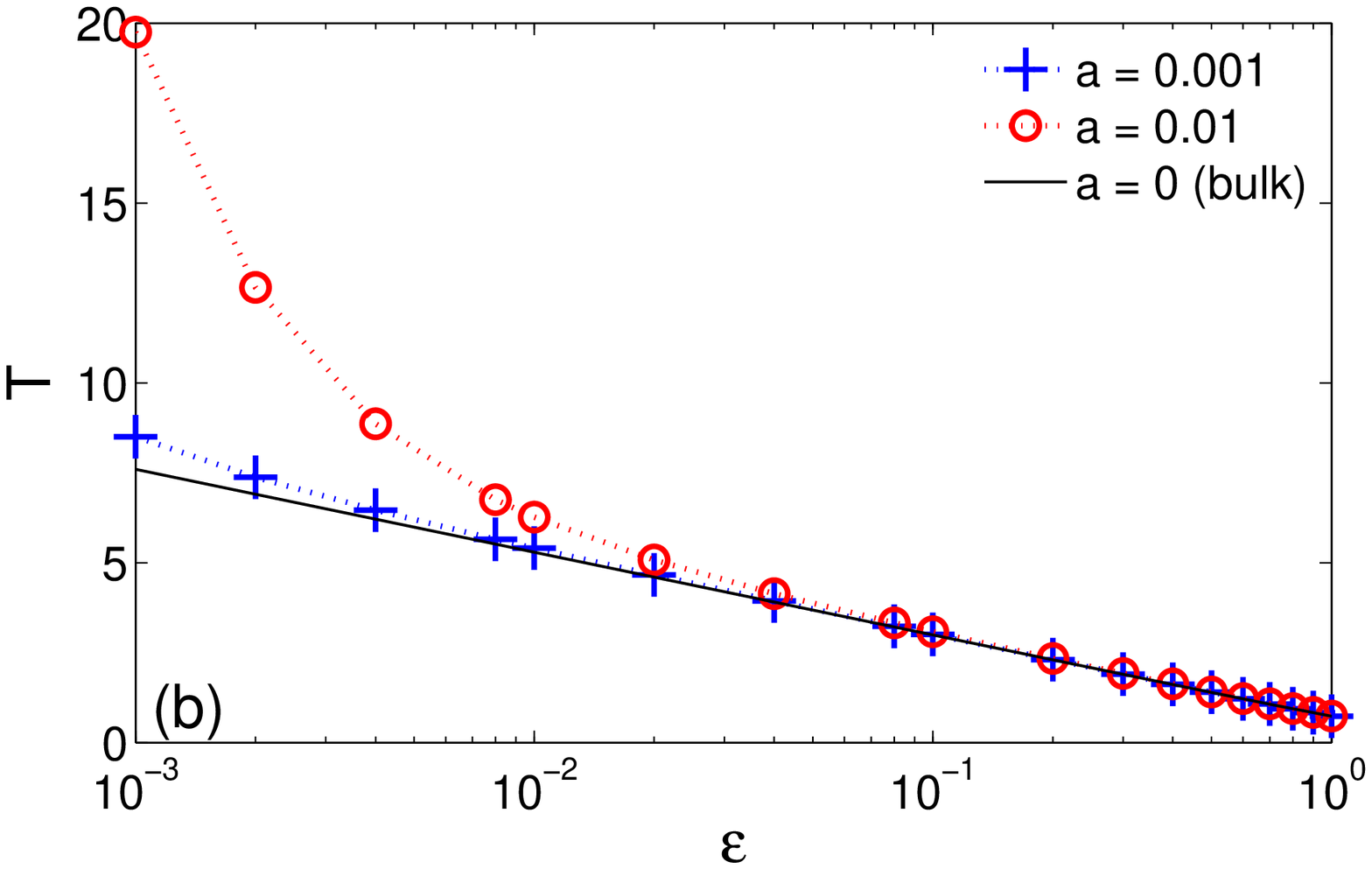}
\end{center}
\caption{
(a) The limiting mean exit time $\L$ as a function of $a$ for
$\epsilon = 0.01$ (circles) and $\epsilon = 0.1$ (crosses), $D_2 = 1$.
Two horizontal lines indicate the mean exit time for pure bulk diffusion
$\langle t_1\rangle_{\lambda = \infty}$: $5.2929$ and $2.9949$,
respectively.  As expected, $\L$ increases with $a$ because every
reflection from the boundary by distance $a$ requires additional time
to return (given that the target is located at the boundary).  In the
limit $a = 1$, $\L$ becomes $78.2898$ for $\epsilon = 0.01$ and
$7.6040$ for $\epsilon = 0.1$ according to Eq. \eqref{eq:T_a=1}.  (b)
$\L$ as a function of $\epsilon$ for $a = 0.001$ (pluses) and $a =
0.01$ (circles).  Solid line shows the mean exit time for pure bulk
diffusion from Eq. \eqref{eq:t1_bulk}.  As expected, the limiting time
diverges logarithmically with $\epsilon$ when $\epsilon \gtrsim a$.
In the opposite case $\epsilon \lesssim a$, the reflection distance
$a$ alters this behavior.}
\label{fig:Tinf}
\end{figure}

\section{Diagonal approximation}

In \cite{Benichou11}, an explicit diagonal approximation to the mean
exit time has been proposed and checked numerically to be accurate.
In our notations, this approximation consists in removing all
nondiagonal elements of the matrix $\VV\TT\VV$ from Eq.
\eqref{eq:VTV} that represents the operator $V\tilde{T}V$ in the
cosine basis.  Under this approximation, one gets
\begin{equation}
\lambda_n \approx \frac{1-(1-a)^n}{n^2} ~ \frac{\pi - \epsilon + \frac{\sin 2n\epsilon}{2n}}{\pi},  
\qquad e_n \approx \sqrt{\frac{2}{\pi}} \cos(n\theta) ,
\end{equation}
from which
\begin{equation}
\label{eq:psin_approx}
\psi_n^2 \approx \frac{2(1-(1-a)^n)}{\pi n^4}\biggl((\pi-\epsilon)\cos n\epsilon + \frac{\sin n\epsilon}{n}\biggr)^2 .
\end{equation}
Substituting these approximations into Eq. \eqref{eq:<t1>2}, we
retrieve the diagonal approximation from \cite{Benichou11}:
\begin{equation}
\label{eq:t1_approx}
\begin{split}
\langle t_1\rangle & \approx \frac{1}{\pi D_1}\left(1 + \lambda \frac{1-(1-a)^2}{4D_2}\right)\biggl[\frac{(\pi-\epsilon)^3}{3} \\
& - \frac{2\lambda}{\pi D_1} \sum\limits_{n\geq 1} \frac{1-(1-a)^n}{n^2} 
\frac{\bigl((\pi-\epsilon)\cos n\epsilon + \frac{\sin n\epsilon}{n}\bigr)^2}{n^2 + \frac{\lambda}{\pi D_1} (1-(1-a)^n)
\bigl(\pi - \epsilon + \frac{\sin 2n\epsilon}{2n}\bigr)} \biggr]. \\
\end{split}
\end{equation}

The diagonal approximation has been shown numerically to be very
accurate for moderate $\lambda$ and arbitrary $\epsilon$ (see
\cite{Benichou11}).  Although the matrix $\VV \TT\VV$ has indeed the
dominant diagonal, the remarkable quality of this approximation
remained puzzling.  At the same time, the approximation becomes
inappropriate in the limit of large $\lambda$.  In order to get a
finite limit of the mean exit time for $\epsilon > 0$ as
$\lambda\to\infty$, the constant term $(\pi-\epsilon)^3/3$ has to be
compensated.  In other words, writing Eq. \eqref{eq:t1_approx} as
\begin{equation}
\label{eq:t1_approx2}
\begin{split}
\langle t_1\rangle & \approx \frac{1}{\pi D_1}\left(1 + \lambda \frac{1-(1-a)^2}{4D_2}\right)\biggl\{\biggl[\frac{(\pi-\epsilon)^3}{3}
- 2 \sum\limits_{n\geq 1} \frac{\bigl((\pi-\epsilon)\cos n\epsilon + \frac{\sin n\epsilon}{n}\bigr)^2}{\pi - \epsilon + \frac{\sin 2n\epsilon}{2n}}
\frac{1}{n^2} \biggr] \\
& + 2 \sum\limits_{n\geq 1} \frac{\bigl((\pi-\epsilon)\cos n\epsilon + \frac{\sin n\epsilon}{n}\bigr)^2}{\pi - \epsilon + \frac{\sin 2n\epsilon}{2n}}
\frac{1}{n^2 + \frac{\lambda}{\pi D_1} (1-(1-a)^n)
\bigl(\pi - \epsilon + \frac{\sin 2n\epsilon}{2n}\bigr)} \biggr\}, \\
\end{split}
\end{equation}
the constant term in square brackets should be canceled.  This is true
only for $\epsilon = 0$ and $\epsilon = \pi$.  As a consequence, the
diagonal approximation \eqref{eq:t1_approx} (or
\eqref{eq:t1_approx2}) implies the divergence of the mean exit time as
$\lambda\to\infty$, in contrast to the expected convergence to a
finite limit for $\epsilon > 0$.

This discrepancy can be interpreted in terms of the asymptotic
behavior of the spectral weights $\psi_n^2$.  In fact, the approximate
relation \eqref{eq:psin_approx} suggests $\psi_n^2 \propto n^{-4}$, in
contrast to the observed $n^{-6}$ decay.  In other words, the diagonal
approximation operates over an intermediate asymptotic regime, while
the ultimate $n^{-6}$ asymptotics is missed.  This observation
explains why the diagonal approximation accurately captures the
behavior of the mean exit time $\langle t_1\rangle$ over small and
moderate $\lambda$ but fails in the limit $\lambda\to\infty$ (see
Fig. \ref{fig:t1_eps}b).

\section*{Conclusion}

In summary, we presented a spectral approach for computing the mean
exit time $\langle t_1\rangle$ from the unit disk by surface-mediated
diffusion with two alternating phases of bulk and surface motion.  The
derived spectral representation is particularly suitable for the
asymptotic analysis of the mean exit time in the limit of large
desorption rates $\lambda$.  For a point-like target ($\epsilon = 0$),
we proved the asymptotic increase $\langle t_1\rangle \propto
\sqrt{\lambda}$ and provided lower and upper bounds.  For extended
targets ($\epsilon > 0$), we investigated the asymptotic approach of
the mean exit time to a finite limit and proved an increase of
$\langle t_1\rangle$ at large $\lambda$ (the first correction term
being $\propto \lambda^{-1/2}$).  We revealed different asymptotic
behaviors of the governing eigenvalues $\lambda_n$ and spectral
weights $\psi_n^2$ depending on the target size $\epsilon$ and the
reflection distance $a$.  Finally, we discussed the optimality of
surface-mediated diffusion and potential limitations of the diagonal
approximation for the mean exit time.


\section{Appendices}
\begin{appendices}
\section{Asymptotic behavior of the eigenvalues of the operator $V\tilde{T}V$}
\label{sec:asympt}

\begin{theorem}
\label{thm:behavior1a}
Let $\lambda_n$ be the eigenvalues of the operator $V\tilde{T}V$,
where the operators $\tilde{T}$ and $V$ are defined in
Eqs. \eqref{eq:eq_T_tilde} and \eqref{eq:operator_V_A}.  We have
\begin{equation}
\lambda_n \sim A_{\epsilon} n^{-2} ,
\end{equation}
where $A_{\epsilon}$ depends only on $\epsilon$.
\end{theorem}
\begin{proof}
In order to prove this statement, we first investigate the following
problem: 

Let $A$ and $B$ are two compact, positive, self-adjoint operators.  We
assume that the eigenvalues of the operator $A$ are ordered in a
decreasing sequence: $\lambda_1(A)\geq \lambda_2(A)\geq \ldots \geq
\lambda_n(A)\geq \ldots \geq 0$.  We recall the variational principle
as following
\begin{theorem}
\begin{equation}
\lambda_n(A)=\max_{F} \min_{O\neq x\in A}\frac{\langle Ax,x\rangle}{\langle x,x\rangle} ;\text{(}\text{where } 
F\text{ is a subspaces of }L^2[0,\pi];\text{ } \dim F=n \text{)}.
\end{equation}
The $\max$ is taken over $F$, the subspace associated with the first
$n$ eigenvectors of $A$.
\end{theorem}
We state two following lemmas which will be needed to prove the
proposition \ref{thm:behavior1a}.

\begin{lemma}
\label{lem:lemma1}
We make the assumption that $\lambda_n(A)$, $\lambda_n(A+B)$ are the
$n^{th}$ eigenvalues of the operators $A$ and $A+B$.  Then, we have
\begin{equation}
\lambda_n(A)-\|B\|\leq \lambda_n(A+B)\leq \lambda_n(A)+\|B\|,
\end{equation}
where $\|.\|$ define the norm of an operator in $L^2[0,\pi]$ space.
\end{lemma}
\begin{proof}
Let $F$ be the subspace of $L^2[0,\pi]$ associated with the first
$n$ eigenvectors of $A$.  For all $x\in F$, we have
\begin{align*}
\langle(A+B)x,x\rangle = \langle Ax,x\rangle + \langle Bx,x\rangle.
\end{align*}
According to the variational principle, we have
\begin{align}\label{eq:apply_Variational_Principle_1_A}
\langle Ax,x\rangle\geq \min_{x\in A}\frac{\langle Ax,x\rangle}{\langle x,x\rangle}\|x\|^2=\lambda_n(A)\|x\|^2.
\end{align}
Besides, we have 
\begin{align}\label{eq:apply_Variational_Principle_2_A}
| \langle Bx,x\rangle | \leq \|B\|\|x\|^2.
\end{align}
It follows from Eqs. \eqref{eq:apply_Variational_Principle_1_A} and
\eqref{eq:apply_Variational_Principle_2_A} that
\begin{align*}
\langle (A+B)x,x\rangle\geq \lambda_n(A)\|x\|^2-\|B\|\|x\|^2.
\end{align*}
This gives
\begin{align*}
\min_{0\neq x\in F}\frac{\langle (A+B)x,x\rangle}{\|x\|^2}\geq \lambda_n(A)-\|B\|.
\end{align*}
Again, according to the variational principle, we thus get
\begin{align*}
\lambda_n(A+B)\geq\lambda_n(A)-\|B\|.
\end{align*} 
In the same manner, if we take $F$ be associated to the first $n$
eigenvectors of $A+B$, we can get
\begin{align*}
\lambda_n(A+B)\leq \lambda_n(A)+\|B\|,
\end{align*}
and the lemma \ref{lem:lemma1} follows.
\end{proof}

\begin{lemma}
\label{lem:lemma2}
With the notations used in lemma \ref{lem:lemma1}, if $\rank(B) <
\infty$, then
\begin{equation}
\label{eq:inequality_A+B}
\lambda_{n+\rank(B)}(A)\leq \lambda_n(A+B)\leq \lambda_{n-\rank(B)}(A).
\end{equation}
\end{lemma}
\begin{proof}
We call $F$ the subspace of $L^2[0,\pi]$ associated with the first
$n-\rank(B)$ eigenvectors of $A$.

By the variational principle, we have 
\begin{equation*}
\forall x\in F\cap \ker(B), \|x\|=1: \langle (A+B)x,x\rangle = \langle Ax,x\rangle\geq \lambda_{n+\rank(B)}(A).
\end{equation*}
Consequently,
\begin{equation*}
\min_{x\in F\cap \ker(B);\\ \|x\|=1}\langle (A+B)x,x\rangle\geq \lambda_{n+\rank(B)}(A).
\end{equation*}
Since 
\begin{equation*}
\dim(F\cap \ker(B))=\dim F -\dim B(F)\geq n,
\end{equation*}
we have 
\begin{equation*}
\lambda_n(A+B)\geq \min_{0\neq x\in \ker(B) \cap F;\\ \|x\|=1}\langle (A+B)x,x\rangle.
\end{equation*}
So, we conclude that
\begin{equation}\label{eq:lemma_2_A}
\lambda_n(A+B)\geq \lambda_{n+\rank(B)}(A).
\end{equation}
The second inequality in \eqref{eq:inequality_A+B} of this lemma is
obtained when we put $A'=A+B$, $B'=-B$, $n'=n-\rank(B)$ and apply the
conclusion \eqref{eq:lemma_2_A} for $A'$, $B'$ and $n'$ instead of
$A$, $B$ and $n$.
\end{proof}

We now call $\pi_N$ be the orthogonal projection on the first $N$
eigenvectors of $B$.

By the property of an orthogonal projection, we can rewrite
\begin{equation*}
B=\pi_N B\pi_N+(I-\pi_N)B(I-\pi_N),
\end{equation*}
then
\begin{equation*}
A+B=A+\pi_N B\pi_N+(I-\pi_N)B(I-\pi_N).
\end{equation*}
We note that 
\begin{align*}
\rank(\pi_N B\pi_N)=N,
\end{align*} 
and
\begin{align}\label{eq:norm_B_is_dominated}
\|(I-\pi_N)B(I-\pi_N)\|\leq \lambda_N(B).
\end{align}
By applying lemma \ref{lem:lemma1}, we get $\forall~ n\geq N$,
\begin{align*}
& \lambda_n(A+\pi_N B\pi_N)-\|(I-\pi_N)B(I-\pi_N)\|\leq\lambda_n(A+B) , \\
& \lambda_n(A+B) \leq\lambda_n(A+\pi_N B\pi_N)+\|(I-\pi_N)B(I-\pi_N)\|. 
\end{align*}
From \eqref{eq:norm_B_is_dominated}, we obtain
\begin{align*}
\lambda_n(A+\pi_N B\pi_N)-\lambda_N(B)\leq\lambda_n(A+B)\leq\lambda_n(A+\pi_N B\pi_N)+\lambda_N(B).
\end{align*}
According to lemma \ref{lem:lemma2},
\begin{align*}
\lambda_{n+N}(A)-\lambda_N(B)\leq\lambda_n(A+B)\leq\lambda_{n-N}(A)+\lambda_N(B).
\end{align*}
We can thus conclude that
\begin{align}\label{eq:eq_eigenvalues}
\forall N, \text{ }\forall n\geq N: \text{ } \lambda_{n+N}(A)-\lambda_N(B)\leq\lambda_n(A+B)\leq\lambda_{n-N}(A)+\lambda_N(B).
\end{align}

\begin{lemma}
\label{lem:lemma3}
Let $\{\lambda_n(A)\}$ and $\{\lambda_n(B)\}$ are the eigenvalues of
two self-adjoint operators $A$ and $B$.  If $\lambda_n(A)\sim c
n^{-s}$ and $\lambda_N(B)=\rho^N$ where $A$ and $\rho$ are some
constants, $0 \leq \rho \leq 1$, then
\begin{align*}
\lambda_n(A+B)\sim c n^{-s}.
\end{align*} 
\end{lemma}
\begin{proof}
Indeed, by applying \eqref{eq:eq_eigenvalues}, if we choose
$N(n)=n^{\delta}$, with $\delta < 1$, then we obtain from
\eqref{eq:eq_eigenvalues} that
\begin{align*}
c (n+n^{\delta})^{-s}(1+o(1))-\rho^{n^{\delta}}\leq \lambda_n(A+B)\leq c(n-n^{\delta})^{-s}(1+o(1))+\rho^{n^{\delta}},
\end{align*}
Let $n \rightarrow \infty$, we have
\begin{align*}
c n^{-s}(1+o(1))\leq \lambda_n(A+B)\leq c n^{-s}(1+o(1)).
\end{align*}
Hence, we conclude that $\lambda_n(A+B)\sim c n^{-s}$.
\end{proof}

We now turn back to prove the Proposition \ref{thm:behavior1a}.

We consider the eigenpairs of the operator $V\tilde{T}V$, where
$\tilde{T}$ is defined in Eq. \eqref{eq:eq_T_tilde}, and $V$ is
defined by
\begin{align*}
V(\cos nx)=\sqrt{1-(1-a)^n}\cos nx\simeq \left(1-\frac{1}{2}(1-a)^n\right)\cos nx
\end{align*}
(the last approximate equality is valid for $a\ll 1$).  We note that
$V$ and $\tilde{T}$ are self-adjoint operators.  We rewrite $V$ as
$I-R$ where
\begin{align*}
R(\cos nx)=\frac{1}{2}(1-a)^n\cos nx,
\end{align*} 
then, 
\begin{align*}
V\tilde{T}V=(I-R)\tilde{T}(I-R)=\tilde{T}-R\tilde{T}-\tilde{T}R-R\tilde{T}R.
\end{align*}
Let us denote by $K_N$ the image of the orthonormal projection on the
first $N^{\rm th}$ eigenvectors of $R$ and $R_N$ the image of the
orthonormal projection on the rest eigenvectors of $R$.  By
definition, $R=K_N+R_N$.  Then,
\begin{align}\label{eq:analys_VTV}
V\tilde{T}V &=\tilde{T}-(K_N+R_N)\tilde{T}-\tilde{T}(K_N+R_N)+(K_N+R_N)\tilde{T}(K_N+R_N)\nonumber\\
&=\tilde{T}\underbrace{-K_N\tilde{T}-\tilde{T}K_N+K_N\tilde{T}K_N+K_N\tilde{T}R_N+R_N\tilde{T}K_N}_{ 
\text{these operators have the finite rank, which equal to }N}-R_N\tilde{T}-\tilde{T}R_N+R_N\tilde{T}R_N.
\end{align}
We note that $\rank(K_N)=N$ and in formula \eqref{eq:analys_VTV},
whenever there is a $K_N$, we have an operator of rank $N$.  Moreover,
$-R_N\tilde{T}-\tilde{T}R_N+R_N\tilde{T}R_N$ has the norm dominated by
the $N^{\rm th}$ eigenvalue of $R$:
\begin{align*}
\|-R_N\tilde{T}-\tilde{T}R_N+R_N\tilde{T}R_N\|&\leq \|R_N \tilde{T}\|+\|\tilde{T} R_N\|+\|R_N\tilde{T}R_N\|\\
&\leq c \|R_N\| \leq c \lambda_N(R)=\frac{c}{2}(1-a)^N.
\end{align*}
Since the operator $\tilde{T}$ is the solution of the Sturm-Liouville
problem, $\lambda_n(\tilde{T})\sim A_{\epsilon} n^{-2}$.  Hence, refer
to lemma \ref{lem:lemma3}, we get that
\begin{align*}
\lambda_n(V\tilde{T}V)\sim A_{\epsilon} n^{-2}.
\end{align*}
\end{proof}


\section{Numerical computation of spectral characteristics}
\label{sec:numerics}

We briefly present a numerical algorithm to compute the spectral
characteristics $\lambda_n$ and $\psi_n$.  In order to compute the
eigenvalues $\lambda_n$ and the eigenvectors $e_n$ of the operator
$V\tilde{T}V$, we get an explicit representation of this operator in
the basis $\cos n\theta$.  First, we find
\begin{equation}
\tilde{T}(\cos n\theta)=
\begin{cases}
	\frac{\cos n\theta-\cos n(\pi-\epsilon)}{n^2}, & 0 \le \theta < \pi-\epsilon,\\
	0, & \pi-\epsilon \le \theta \le \pi,
\end{cases}
\end{equation}
and
\begin{equation}
\tilde{T}(1)(\theta)=
\begin{cases}
	\frac{(\pi-\epsilon)^2-\theta^2}{2}, & 0 \le \theta < \pi-\epsilon,\\
	0, & \pi-\epsilon \le \theta \le \pi,
\end{cases}
\end{equation}
from which the expansion of $\tilde{T}(\cos n\theta)$ ($n\geq 0$) in
the basis $\{\cos n\theta\}$ of $L_{\text{even}}^{2}{[0,\pi]}$ is
\begin{equation}
\tilde{T}(\cos n\theta) = \sum_{m\geq 0} \TT_{mn}\cos m\theta    \quad n\geq 0 ,
\end{equation}
where the coefficients $\TT_{mn}$ are defined by
\begin{align}
\frac{\pi}{2} \TT_{mn}&=
	\begin{cases}
	\langle \tilde{T}(\cos n\theta),\cos m\theta\rangle, & \text{ if } m\geq 1\\
	\frac{1}{2}\langle \tilde{T}(\cos n\theta),1\rangle, & \text{ if } m=0
		\end{cases}\nonumber\\
&=
	\begin{cases}
	\frac{1}{2}(1-\delta_{mn}) \frac{(-1)^{m+n+1}}{mn} \biggl[\frac{\sin (m-n)\epsilon}{m-n} - \frac{\sin (m+n)\epsilon}{m+n}\biggr]
        + \frac{1}{2} \delta_{mn}\frac{1}{n^2}\left(\pi-\epsilon + \frac{\sin 2n \epsilon}{2n}\right) \\
	\hspace{80mm} (m\geq 1;~ n\geq 1), \\
	\frac{1}{2} \frac{(-1)^{n+1}}{n^2}\left[(\pi-\epsilon)\cos n\epsilon + \frac{\sin n\epsilon}{n}\right], \hspace{30mm} (m=0;~ n\geq 1),  \\
	\frac{(-1)^{m+1}}{m^2}\left[(\pi-\epsilon)\cos m\epsilon + \frac{\sin m\epsilon}{m}\right], \hspace{28mm} (m\geq 1;~ n=0),  \\
	\frac{(\pi-\epsilon)^3}{6}, \hspace{69.5mm}  (m=0;~ n=0).
	\end{cases}
\end{align}
In turn, the operator $V$ has a diagonal representation:
\begin{align}
\VV_{mn}=
	\begin{cases}
		\sqrt{1-(1-a)^m} &  (m=n;~ m,n\geq 0), \\
		0 & (m\neq n;~ m,n\geq 0) . 
	\end{cases}
\end{align}
Combining these results, the operator $V\tilde{T}V$ is represented by
the infinite-dimensional matrix $\VV \TT \VV$ whose elements are
\begin{align}
\label{eq:VTV}
[\VV \TT \VV]_{m,n} & = \frac{1}{\pi} \sqrt{1-(1-a)^n} ~ \sqrt{1-(1-a)^m}\biggl\{
   \delta_{mn}\frac{1}{n^2}\left(\pi-\epsilon + \frac{\sin 2n \epsilon}{2n}\right)  \\
& - (1-\delta_{mn}) \frac{(-1)^{m+n}}{mn} 
\biggl[\frac{\sin (m-n)\epsilon}{m-n} - \frac{\sin (m+n)\epsilon}{m+n}\biggr] \biggr\} \quad (m\geq 1;~ n\geq 1),  \nonumber
\end{align}
and $[\VV \TT \VV]_{m,n} = 0$ if $m = 0$ or $n = 0$.  Solving the
eigenvalues $\{\lambda_n\}$ and the eigenvectors $\{e_n\}$ of the
operator $V\tilde{T}V$ is equivalent to finding the eigenpairs of the
associated matrix $\VV \TT \VV$.  Note that this matrix is symmetric.

The matrix $\VV \TT \VV$ is diagonalized in Matlab that finds the
eigenvalues $\lambda_n$ and the coefficients $v_{mn}$ determining the
orthonormal basis $\{e_n\}_{n \geq 0}$ as
\begin{equation} \label{eq:eq_of_e_n}
e_n(\theta) = \sqrt{\frac{1}{\pi}}~ v_{0n} + \sqrt{\frac{2}{\pi}} \sum_{m\geq 1} v_{mn} \cos m\theta.
\end{equation}
The spectral weights $\psi_n$ are then given as
\begin{align}
\psi_n = \langle \psi, e_n \rangle = \sqrt{\frac{2}{\pi}} \sum_{m\geq 1} v_{mn} \langle \psi, \cos m\theta\rangle ,
\end{align}
where
\begin{align*}
\langle \psi,\cos m\theta \rangle =
	\sqrt{1-(1-a)^m}~\frac{(-1)^{m+1}}{m^2}\left[(\pi-\epsilon)\cos m\epsilon + \frac{\sin m\epsilon}{m}\right],  
\end{align*}
and $\langle \psi,1 \rangle = \langle V\tilde{T}(1), 1 \rangle =
\langle \tilde{T}(1), V1 \rangle = 0$.

\section{Pure bulk diffusion phase}
\label{sec:bulk}

In \cite{Rupprecht14}, the mean exit time for pure bulk diffusion phase
($\lambda = \infty$, $a = 0$) was found to be
\begin{equation}
t_2(r,\theta) = \frac{1}{D_2} \left(\frac{1-r^2}{4} + \frac{\alpha_0}{2} + \sum\limits_{n=1}^\infty \alpha_n r^n \cos n\theta \right),
\end{equation}
where
\begin{align}
\alpha_0 & = - 2\ln [\sin (\epsilon/2)] , \\
\alpha_n & = \frac{(-1)^{n-1}}{2n} \bigl[P_n(\cos \epsilon) + P_{n-1}(\cos\epsilon)\bigr], 
\end{align}
where $P_n(z)$ are Legendre polynomials.  In the limit $r\to 1$, one
gets
\begin{equation}
t_1(\theta) = \begin{cases} \frac{1}{D_2} \left(\frac12 \alpha_0 + \sum\limits_{n=1}^\infty \alpha_n  \cos n\theta \right) 
\quad (0 < \theta < \pi - \epsilon),  \cr 
  0 , \hskip 39mm  (\pi - \epsilon \leq \theta \leq \pi ) .  \end{cases}
\end{equation}
The average of the uniformly distributed starting point $\theta$
yields
\begin{equation}
\label{eq:t1_bulk}
\langle t_1\rangle_b = \frac{1}{\pi D_2} \left(-(\pi-\epsilon)\ln [\sin(\epsilon/2)] + 
\sum\limits_{n=1}^\infty \frac{\sin n\epsilon}{2n^2} \bigl[P_n(\cos \epsilon) + P_{n-1}(\cos\epsilon)\bigr] \right) .
\end{equation}
For small $\epsilon$, the first term dominates yielding $\langle
t_1\rangle_b \simeq \frac{\ln(2/\epsilon)}{D_2} (1 + O(\epsilon))$.

Figure \ref{fig:t1_surfbulk} shows the mean exit times $\langle
t_1\rangle_{\lambda =0}$ and $\langle t_1\rangle_{\lambda = \infty}$
from Eqs. \eqref{eq:t1_surface} and \eqref{eq:t1_bulk} for surface
diffusion and pure bulk diffusion, as a function of $\epsilon$.

\begin{figure}
\begin{center}
	\includegraphics[width=0.49 \textwidth]{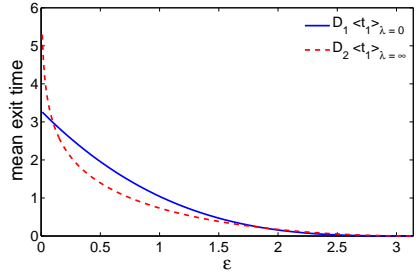}
\end{center}
\caption{
Mean exit times $\langle t_1\rangle_{\lambda =0}$ and $\langle
t_1\rangle_{\lambda = \infty}$ from Eqs. \eqref{eq:t1_surface} and
\eqref{eq:t1_bulk} for surface diffusion and pure bulk diffusion, as a
function of $\epsilon$.  These times are multiplied by the
corresponding diffusion coefficients $D_1$ and $D_2$.}
\label{fig:t1_surfbulk}
\end{figure}

\section{Transportation case ($a=1$)}
\label{sec:transportation}

As we earlier discussed, one typically considers small values of the
ejection distance $a$.  Nevertheless, the above results are applicable
to any value of $a$ from $0$ to $1$.  The so-called transportation
case $a = 1$ when the particle is reflected to the origin, was studied
by B\'enichou {\it et al.} \cite{Benichou11}.  In this case,
successive explorations between any two reflections are independent
that allows one to get much simpler formulas.  For instance, in the
limit $\lambda\to\infty$, the Laplace transformed probability density
of the exit time, ${\mathcal L}[P(t)](s)$ (with a uniformly chosen
initial point on the circle), has a simple expression: ${\mathcal
L}[P(t)](s) = q\bigl[1 -
(1-q)/I_0\bigl(\sqrt{s/D_2}\bigr)\bigr]^{-1}$, where $q =
\epsilon/\pi$, and $1/I_0(\sqrt{s/D_2})$ is the Laplace transformed
probability density for the first passage time to the circle when
started from the origin (with $I_0(z)$ being the modified Bessel
function of the first kind).  As a consequence, the mean exit is
simply
\begin{equation}  \label{eq:T_a=1}
\L_{a=1} = \frac{\pi - \epsilon}{4D_2 \epsilon} .
\end{equation}
This limit is clearly seen on Fig. \ref{fig:Tinf}.

\end{appendices}



\bibliographystyle{plain}
\bibliography{article}

\end{document}